\documentclass[12pt]{article}
\usepackage{setspace}
\usepackage{color}
\usepackage{amsfonts}
\usepackage{amssymb}
\usepackage{amsmath}
\usepackage[mathscr]{eucal}
\usepackage{amsfonts}
\usepackage{amsmath,amsxtra,latexsym,amsthm, amssymb, amscd}
\usepackage[colorlinks,citecolor=blue,filecolor=black,linkcolor=blue,urlcolor=black]{hyperref}
\usepackage{pgfplots} %to plot
\usepackage{url}
\usepackage[round]{natbib}
\usepackage{fancyhdr}
\usepackage{graphicx}
\usepackage{hyperref}

\usepackage{caption}
\usepackage{subcaption}
\usepackage{epstopdf}

\usepackage{soul}

\usepackage{titlesec}

\usepackage[T1]{fontenc}
\usepackage[mathscr]{eucal}
\usepackage{amsfonts}
\usepackage{amsmath,amsxtra,latexsym,amsthm, amssymb, amscd}
\usepackage[round]{natbib}
\usepackage{fancyhdr}
\usepackage{graphicx}
\usepackage{lipsum}		

\usepackage{filecontents}

\usepackage[latin1]{inputenc}

\usepackage{lmodern}
\usepackage[none]{hyphenat} %éviter la coupure des mots.
\newtheorem{ex}{Example}
\newcommand{\bex}{\begin{ex}}
\newcommand{\eex}{\end{ex}}
 \newtheorem{remark}{Remark}
\usepackage{fancyhdr}
\newtheorem{theorem}{Theorem}
\newtheorem{definition}{Definition}
\newtheorem{lemma}{Lemma}

\newtheorem{proposition}{Proposition}
\newtheorem{assum}{Assumption}
\newtheorem{corollary}{Corollary}

\newcommand{\rr}{\mathbb{R}}

\newcommand{\ee}{\mathcal{E}}

\newcommand{\summ}{\sum\limits}   
   
\newcommand{\ma}{\max\limits}

\newcommand{\limm}{\lim\limits}

\begin{document}

\title{The relationship between general equilibrium models with infinitely-lived agents and overlapping generations models, and some applications\footnote{I would like to thank Stefano Bosi, Cuong Le Van, Alexis Akira Toda, an associate editor and two anonymous referees for their helpful comments and discussions.}}% in asset price bubble and equilibrium indeterminacy}% {\small (There are some english errors - please do not circulate without permission) 
%}
\author{Ngoc-Sang PHAM\thanks{Emails: npham@em-normandie.fr and pns.pham@gmail.com.  Phone:  +33 2 50 32 04 08. Address: EM Normandie (campus Caen), 9 Rue Claude Bloch, 14000 Caen, France. }\\
EM Normandie Business School, M\'etis Lab (France) }
\date{\today}
\maketitle

\begin{abstract}We prove that a two-cycle equilibrium in a general equilibrium model with infinitely-lived agents (GEILA) constitutes an equilibrium in an overlapping generations (OLG) model. Conversely, an equilibrium in an OLG model that satisfies additional conditions is part of an equilibrium in a GEILA model. Our framework, which includes three assets (physical capital, a Lucas tree, and fiat money), encompasses both exchange and production economies. As an application, we demonstrate that equilibrium indeterminacy and rational asset price bubbles can arise not only in OLG models but also in GEILA models.
%Note that  our models consisting of three assets (physical capital, Lucas' tree, and fiat money) cover both exchange and production economies. Applying this result, we demonstrate that equilibrium indeterminacy and rational asset price bubbles may arise not only in OLG models but also in models with infinitely-lived agents.  

\end{abstract}

%\section{Equation}

%\begin{align}\sqrt{x^2+3x}+2\sqrt{x-1}=2x+\sqrt{\frac{x^2+2x-3}{x}}
%\end{align} 

\section{Introduction}

General equilibrium models with infinitely-lived agents (GEILA)  and overlapping generations (OLG) models are two workhorses in macroeconomics. A vast body of literature explores these two frameworks.\footnote{See \cite{delaCroixMichel2002} for an introduction to OLG models, and \cite{Becker2006}, \cite{MagillQuinzii_book2008}, \citet{LeVanPham2016}, among others, for an introduction to GEILA models.} This raises a natural question: what is the relationship between these two classes of models? If so, can this relationship help us to address economic questions?

%Surprisingly, the existing literature has not been provided a convincing answers to these questions. 
Looking back to history,  \cite{woodford86} considered an economy with capital accumulation and money, where there are two classes of infinitely-lived agents (capitalists and workers). \cite{woodford86} studied a special setup in which capitalists have logarithmic utility, never hold money, and face a single trade-off between consumption and investment. Workers, on the other hand, never purchase capital and face a trade-off between consumption and leisure. Then, he obtained an equilibrium system, which is similar to those in an OLG model with two-period-lived workers.\footnote{Budget constrains (1.1b) in \cite{woodford86} writes $p_t\big((c^w_t+(k^w_t-dk^w_{t-1})\big)+M^w_{t+1}=M^w_{t}+r_tk^w_{t-1}+w_tn_t$. He also imposes constraints $k^w_t\geq 0$, $M^w_{t+1}\geq 0$, and borrowing constraint $p_t\big((c^w_t+(k^w_t-dk^w_{t-1})\big)\leq M^w_{t}+r_tk^w_{t-1}$. He focuses on the case in which workers choose $k^w_t=0$ for any $t$ in optimal.} 
Following \cite{woodford86}, \cite{Kocherlakota1992} wrote, in his footnote 4, that "In both examples, short sales constraints that bind in alternating periods serve to make the infinite-horizon economy look like an overlapping generations economy". 

To date, neither \cite{woodford86}, \cite{Kocherlakota1992}, nor the broader literature has formally established a connection between these two classes of models in a general framework. Our paper seeks to address this gap.

Our contribution is two-fold. First, we prove that (1) a two-cycle equilibrium in a GEILA model is also an equilibrium in an OLG model, and (2) conversely, an equilibrium in an OLG model is part of a two-cycle equilibrium in a GEILA model if and only if it satisfies additional conditions including the transversality conditions. %Our second contribution is to show some applications of this observational equivalence in several topics including, equilibrium indeterminacy and rational asset price bubbles.

Compared to \cite{woodford86} and \cite{Kocherlakota1992}, we establish the  observational connection in more general frameworks (including general utility functions, general endowments, and multiple assets). It should be noted that an equilibrium in an OLG model is not automatically part of an equilibrium in GEILA models. In particular, it is necessary to check the transversality conditions.

The existing literature also highlights a connection between standard OLG models and  dynamic programming frameworks. \citet{Aiyagari1985} demonstrates that the dynamics of capital in a standard OLG model (Diamond's model) can be derived from a discounted dynamic programming framework. \cite{Hou1987} considers pure exchange economies and establishes an observational equivalence between an OLG model with agents living for two periods and a cash-in-advance economy with a single infinitely-lived representative agent. \cite{LovoPolemarchakis2010} depart from a model with an infinitely-lived representative agent and show how the qualitative properties of OLG economies can be replicated by introducing a certain level of myopia.\footnote{It is also known that, in some cases, an OLG model with positive bequests can be reformulated as an optimal growth model \`a la Ramsey (see \cite{Barro1974}, \cite{Aiyagari1992}, \cite{MichelThibaultVidal2006} among others).}

Our paper focuses on general equilibrium models with a finite number of infinitely-lived households, which are more general than models with a single representative household.  Notice that the results in \citet{Aiyagari1985} and \cite{Hou1987} cannot be applied to our models because our framework includes endowments, physical capital, and long-lived assets (both with and without dividends), while the model in \citet{Aiyagari1985} features only physical capital (similar to a one-sector optimal growth model), and \cite{Hou1987} considers an exchange economy.

%\cite{woodford86} wrote that "In the present example, all agents are infinite lived; however, it is shown that that the existence of finance constraints (in addition to each agent's budget constraint) can result in equilibrium dynamics similar to those obtained for overlapping generations models". Then,  \cite{Kocherlakota1992} wrote, in his footnote 4, that "... short sales constraints that bind in alternating periods serve to make the infinite-horizon economy look like an overlapping generations economy". 

%Under Woodford's specifications, solving for equilibrium reduces to solving a two-dimensional difference equation. In contrast, our models may involve a three-dimensional system with infinitely many parameters. Moreover, we work under general utility functions.

%\cite{Kocherlakota1992} also wrote, in his footnote 4, that "... short sales constraints that bind in alternating periods serve to make the infinite-horizon economy look like an overlapping generations economy". However, as \cite{woodford86}, he did not formalize this intuition. Our paper provides a formalization of this idea within a general setting.

As the second contribution, we apply our results to show how equilibrium indeterminacy and rational asset price bubbles can arise in both classes of models.

Our first application concerns equilibrium indeterminacy.  Looking back at history, \citet{KehoeLevine1985} consider two stationary pure exchange economies: the first involves a finite number of infinitely-lived consumers, and the second (an OLG model) features an infinite number of finitely-lived consumers. They argue that these two models have different implications: in the first model, equilibria are generically determinate, whereas this is not the case in the second model.\footnote{See \citet{Farmer2019} for an overview of equilibrium indeterminacy in macroeconomics.} 

%NEW: A conventional view is that OLG models and models with infinitely-lived agents  have different implications ...

The models in our paper are more general than those in \citet{KehoeLevine1985} in the sense that we incorporate capital accumulation and imperfect financial markets (in forms of borrowing constraints). Different from \citet{KehoeLevine1985}, we show that equilibria may be indeterminate in both models.  Precisely, we demonstrate that in a non-stationary exchange economy with a finite number of infinitely-lived consumers, equilibrium indeterminacy can arise. The intuition is that in such an economy, the equilibrium system can be supported by an OLG model, which creates room for indeterminacy.

The second application concerns the issue of rational asset price bubbles which has attracted significant attention from scholars in recent years.\footnote{For detailed surveys, see \cite{BrunnermeierOehmke2012}, \cite{Miao2014}, \citet{MartinVentura2018}, \citet{HiranoTodaJME2024,HiranoToda2025Notes}.} Since \citet{tirole85}, it has become relatively straightforward to build OLG models with bubbles. However, in infinite-horizon general equilibrium models, it is well known that constructing a model where rational asset price bubbles exist is more challenging, particularly when assets yield dividends \citep{tirole82, Kocherlakota1992, SantosWoodford1997}.\footnote{\cite{LeVanPham2016}, \cite{BosiLevanPham2017, BosiLevanPham2017b, BosiLevanPham2018, BosiHaHuyLeVanPhamPham2018, BosiLevanPham2022}  construct models where assets with positive dividends exhibit bubbles. Inspired by \cite{wilson81} and \cite{tirole85} (Proposition 1.c),  \cite{HiranoTodaJPE2025}  construct some models and provide conditions under which any equilibrium (if it exists) is bubbly.} A key difficulty, as proved in \cite{BosiLevanPham2022}'s Proposition 2, is that, in general, the existence of bubbles in such models requires that the asset holdings of at least two agents fluctuate over time and that the borrowing constraints of at least two agents bind at infinitely many periods.\footnote{See also discussions in Section 4.1 in \cite{BosiLevanPham2018}.}

This property leads to the notion of a two-cycle equilibrium in GEILA models, as introduced above (note that this two-cycle structure is the simplest one of the GEILA models that can generate rational asset price bubbles). Building on our observational connection, this two-cycle equilibrium can be supported by an equilibrium in an OLG model.  Thus, if the latter equilibrium exhibits a bubble, we can apply our results and impose additional conditions (which hold under reasonable assumptions) to prove that it is part of a bubbly equilibrium in a GEILA model. 

Thanks to our observational connection, constructing models with infinitely-lived agents where asset bubbles exist is no longer a difficult task. This insight allows us not only to recover but also to extend many models of rational bubbles found in the literature. For instance, Example 1 in
 \cite{Kocherlakota1992} presents an equilibrium where the fiat money has a positive price. However, by applying our result, we go further by showing that, in his model, there exists a continuum of equilibria where the fiat money's price is  strictly positive.

The rest of the paper is organized as follows. Section \ref{twomodels} introduces both GEILA and OLG models. Section \ref{sectionmain} formally establishes the connection between these two models. Section \ref{appli} presents applications of our results to the study of equilibrium indeterminacy and asset price bubbles. Technical proofs are presented in  Appendix \ref{appen}.

%\cite{weil89}, \cite{delaCroixMichel2002}, \citet{tirole85}, \citet{hy17}

\section{Two models} 
\label{twomodels}

\subsection{An overlapping generations model}\label{OLG-1}
We present an OLG framework, which can be considered as a unified model of \citet{tirole85} and \citet{weil90}.\footnote{See \cite{delaCroixMichel2002} for an introduction of OLG models.} This is a discrete time model and the set of times is  $\{0, 1, 2, \ldots\}$. There is a consumption good, which is taken as num\'eraire. 

In each period $t$, there is a representative firm (without market power) that maximizes its profit $\ma_{K_t,L_t\geq 0}\big\{F(K_{t},L_{t})-r_{t}K_{t}-w_{t}L_{t}\big\}$ by choosing the physical capital $K_t$ and the labor $L_t$, where $r_t$ is the rental rate and $w_t$ is the wage. 

The consumer born in period $t$ lives for two periods (young and old) and has  $e_{t}^{y}\geq 0$ units of consumption as endowments at date when young and $e_{t+1}^{o}\geq 0$ when old. Endowments are exogenous. We assume that there is no population growth and the population size $N_t$ on date $t$  is normalized to 1.   %The young at date $t$ can invest by using the physical capital, the long-lived asset bringing dividend (i.e., Lucas' tree) and the pure bubble asset (or fiat money).

This consumer can invest/save using three assets: the physical capital, a long-lived asset that pays dividends (Lucas' tree), and a pure bubble asset.  The structure of the long-lived asset (Lucas' tree) is the following: if the consumer buys $1$ unit of this asset with price $q_t$ on date $t$, she will receive $d_{t+1} $ units of consumption good as dividend and she will be able to resell the asset with price $q_{t+1}$ on  date $t+1$. The positive sequence of real dividends $(d_{t})$ is exogenous. %This asset can be interpreted as land or Lucas' tree  \citep*{Lucas1978}. %security \citep*{SantosWoodford1997},  stock \citep*{Kocherlakota1992}, ... 

Regarding the pure bubble asset (or fiat money), if the consumer buys $a_t$ units of this asset with the price $p_t$ on date $t+1$, then he(she) will resell this asset with the price $p_{t+1}$ on date $t+1$ to receive $p_{t+1}a_t$ units of consumption good. As in the traditional literature \citep{tirole85},  the only reason why people buy this asset is to be able to resell it in the future.
%we assume that there is another long-lived asset with a structure similar to Lucas' tree, but this asset does not bring any dividend. We refer this asset "fiat money"  or "pure bubble asset" as in the traditional literature.

Households born at date $t\geq 0$ choose consumptions $c_{t}^{y},c_{t+1}^{o}$, investment in physical capital $s_{t}$, investment in a long-lived asset $a_{t}$ (Lucas' tree) and pure bubble asset $b_t$ in order to maximize her intertemporal utility 
$u(c_{t}^{y})+\beta u(c_{t}^{o})$ subject to the following constraints
\begin{subequations}
\begin{align}
\label{budget1}& c_{t}^{y}+s_{t}+q_{t}a_{t}+p_tb_t\leq e_{t}^{y}+w_{t},\\%   \\
\label{budget2}& c_{t+1}^{o}\leq e_{t+1}^{o}+(1-\delta+r_{t+1})s_{t}+(q_{t+1}+d_{t+1})a_{t}+p_{t+1}b_t, \\
& s_{t},a_{t},b_t, c_{t}^{y},c_{t}^{o}\geq 0, \notag
\end{align}
\end{subequations}
where $\delta \in [0,1]$ is the depreciation rate of physical capital.

Households born at date $-1$ just consume, that is $c^o_0=e^o_0+(1-\delta+r_0)s_{-1}+(q_0+d_0)a_{-1}+p_0b_{-1}$, where $s_{-1},a_{-1},b_{-1}$ are exogenous.

Denote $R_t\equiv 1-\delta+r_t.$ Let us provide a formal definition of equilibrium.

\begin{definition}\label{def_equilibrium} %Let $a_{-1}=1, b_{-1}=1$, $k_0\geq 0$, $e_{t}^{y}\geq 0, (k_0,e^y_0)\not=(0,0)$, $e_{t}^{o}\geq 0$. 
An intertemporal equilibrium of the two-period OLG economy is a non-negative  list $(s_{t},a_{t}, b_t, c_{t}^{y},c_{t}^{o}, K_t, L_t, w_t, R_t,q_t,p_t)_{t\geq 0}$ satisfying three conditions: 
(1) given $R_{t+1}, q_t, q_{t+1}, p_t, p_{t+1}$ and $w_t$, the list $(s_{t},a_{t}, b_t, c_{t}^{y},c_{t}^{o})_{t\geq 0}$ is a solution to the household's problem and the couple $(K_t,L_t)$ is a solution to the firm's problem,  (2) markets clear: $L_t=1,K_{t+1}=s_t$, $a_t=1$, $b_t=1$ and $c^y_{t}+c^o_{t}+s_{t}=F(K_{t},1)+(1-\delta)K_t+e^y_{t}+e^o_{t}+d_{t}$, and (3) $w_t>0,R_t>0,q_t>0,p_t\geq 0$ $\forall t\geq 0$.

\end{definition}

Our framework encompasses both exchange and production economies. Indeed, a pure exchange economy corresponds to the special case of our model with $\delta = 1$, $s_{-1}=0$, and $F(K,L)=0$ for all $(K,L)\in \mathbb{R}_+^2$ (in this case, we remove variables $s_t,K_t,L_t,w_t,R_t$ and conditions $L_t=1,K_{t+1}=s_t,w_t>0,R_t>0$ from Definition \ref{def_equilibrium}).  Moreover, when $e^y_t = e^o_t = 0$ for all $t$, our framework reduces to the standard overlapping generations model with production and a constant population.

Standard assumptions are required.
\begin{assum}[utility function]
\label{Assumption1} 
The utility function $u: \rr_+\to \rr\cup \{-\infty\}$ is concave, strictly increasing, continuously differentiable and  $u'(0)=+\infty$.
\end{assum}

%The production function $F: \rr_+^{2}\to\rr_+$ is assumed to be constant return to scale (CRS). As usual, we define the function $f: \rr_+\to \rr_+$ by $f(k)\equiv F(k,1)$ $\forall k\geq 0$.
\begin{assum}[production function]\label{Assumption1_production}
The production function $F: \rr_+^{2}\to\rr_+$ is assumed to be constant return to scale (CRS), concave, increasing in each component, continuously differentiable on $(0,\infty)^2$. The function  $f: \rr_+\to \rr_+$, defined by $f(k)\equiv F(k,1)$ $\forall k\geq 0$, is concave, strictly increasing,  continuously differentiable, $f(0)=0$. The depreciation rate $\delta\in [0,1]$.
\end{assum}
%\begin{assum}[endowment]$e_{t}^{y}\geq 0, (k_0,e^y_0)\not=(0,0)$, $e_{t}^{o}\geq 0$.\end{assum}
\begin{assum}[dividends and endowments]\label{Assumption1_dividend}Dividends and endowments satisfy $ 0 < d_t <\infty, e_{t}^{y}\geq 0, e_{t}^{o}\geq 0$ $\forall t$, and $a_{-1}=1, b_{-1}=1$, $s_{-1}\geq 0$, $(s_{-1},e^y_0)\not=(0,0)$. When we consider the exchange economy, we assume that $e^y_t>0$ $\forall t$. 
\end{assum}

To simplify our exposition, let us denote the two-period OLG economy by $\ee_{OLG}\equiv \ee_{OLG}(u,\beta,(e^y_t,e^o_t)_t,f,\delta,(d_t)_t).$

Let us focus on interior equilibria in the sense that $K_t>0 $ $ \forall t$ (this is ensured, for instance, by the Inada condition $f'(0)=+\infty$). In equilibrium, we also have $L_t=1>0$. By consequence,  the first order conditions (FOC) of the firm's problem give 
\begin{align} \label{focfirms}w_{t}=f(K_{t})-K_{t}f^{\prime }(K_{t}) \text{ and } r_{t}=f^{\prime }(K_{t}).
\end{align}
Since  $a_t,b_t>0$ and $s_{t}=K_{t+1}>0$ in any interior equilibrium, we have the following FOCs of households:
\begin{subequations}\label{fochouseholds}
\begin{align}\label{focyo}
u'(c^y_t)&=\beta R_{t+1}u'(c^o_{t+1}),\\
\label{b1}q_{t}R_{t+1}&=q_{t+1}+d_{t+1},\\ 
\label{bd1}p_{t}R_{t+1}&=p_{t+1}.
\end{align}
\end{subequations}
Note also that under conditions  (\ref{b1}), (\ref{bd1}) and $e^y_t+w_t>0$, the list $(s_{t},a_{t}, b_t, c_{t}^{y},c_{t}^{o})_{t\geq 0}$ is a solution to the household's maximization problem if (i) $a_t=b_t=1$, $s_t>0$, (ii) condition (\ref{fochouseholds}) holds, and (iii) budget constraints (\ref{budget1}), (\ref{budget2}) bind.

%If $u(c)=\ln (c)$, we can compute $c_{t}^{y} =w_{t}/(1+\beta )$, $s_{t}+q_{t}a_{t} =\beta w_{t}/(1+\beta )$.
By using market clearing conditions $K_{t+1}=s_{t},  L_{t}=1,  a_{t}=1, b_t=1$, the FOC (\ref{focyo}) can be rewritten as
\begin{align}\label{d1}u'(e_{t}^{y}+w_{t}-K_{t+1}-q_{t}-p_t)&=\beta R_{t+1}u'\big(e_{t+1}^{o}+R_{t+1}(K_{t+1}+q_t+p_t)\big).
\end{align}

To summarize our above arguments, we state the following result.

\begin{lemma}\label{equi-olg} Let Assumptions \ref{Assumption1}-\ref{Assumption1_dividend} be satisfied. %Assume also that $a_{-1}=1, b_{-1}=1, K_0\geq 0, e_{t}^{y}\geq 0, (K_0,e^y_0)\not=(0,0)$,  $e_{t}^{o}\geq 0$.  
A non-negative  list\\
$(s_{t},a_{t}, b_t, c_{t}^{y},c_{t}^{o}, K_t, L_t, w_t, R_t,q_t,p_t)_{t\geq 0}$ is an interior intertemporal equilibrium of the OLG economy if  and only if (1) conditions (\ref{focfirms}), (\ref{b1}), (\ref{bd1}), (\ref{d1})  and  market clearing conditions in Definition \ref{def_equilibrium} are satisfied, (2) the budget constraints (\ref{budget1}) and (\ref{budget2}) bind and $K_{t}>0$ $\forall t\geq 0$. 

%(2) a non-negative list $(q_t,p_t,K_{t+1})_{t\geq 0}$ of asset prices and capital stock, satisfying the following conditions.\footnote{See \cite{BosiHaHuyLeVanPhamPham2018} for an equilibrium analysis for the case $p_t=0,\forall t$.}
%\begin{subequations}\begin{align}\label{d1}
%u'(e_{t}^{y}+f(K_{t})-K_{t}f^{\prime }(K_{t})-K_{t+1}-q_{t}-p_t)&=\beta R_{t+1}u'\big(e_{t+1}^{o}+R_{t+1}(K_{t+1}+q_t+p_t)\big)\\
%\label{b1}q_{t}R_{t+1} &=\left(q_{t+1}+d_{t+1}\right)\\
%p_{t}R_{t+1}&=p_{t+1}\\
%\label{bd1}K_{t+1}>0, q_t&\geq 0, p_t\geq 0\end{align}\end{subequations}
\end{lemma}
%\begin{proof}See Appendix \ref{appen}.\end{proof}
According to Lemma \ref{equi-olg}, an interior equilibrium can be uniquely determined via the sequence $(q_t,p_t,K_{t+1})_{t\geq 0}$. So, we also call $(q_t,p_t,K_{t+1})_{t\geq 0}$  an equilibrium. When we consider an exchange economy, we call  $(q_t,p_t)_{t\geq 0}$ an equilibrium (in this case, we define $R_{t+1}$ by \eqref{b1} instead of $R_{t+1}=1-\delta +r_{t+1}$).

%\begin{remark}\cite{tirole85}'s model with a constant population corresponds to a special case of our model where $d_t=d, e^y_t=e^o_t=0$  $\forall t$. \end{remark}
%
%When $u(c)=\ln(c)$, equation (\ref{d1}) becomes $K_{t+1}+q_{t} = \dfrac{\beta }{1+\beta }\left[ f(K_{t})-K_{t}f^{\prime }(K_{t})\right]$.

\subsection{A general equilibrium model with infinitely-lived agents} \label{GE-1}
We now develop the model in \citet{LeVanPham2016} by adding two ingredients: endowments and a pure bubble asset, allowing us to cover both exchange and production economies.  Consider an infinite-horizon general equilibrium model without uncertainty and discrete time ($t\in \{0,1,2,\ldots\}$). There are  $m$ heterogeneous households and a representative firm without market power. There is a single consumption good, which is the num\'eraire.

For each period $t$, the representative firm  takes prices $(r_t,w_t)$ as given and maximizes its profit by choosing physical capital  $K_t$ and labor $L_t$.
\begin{eqnarray} (P(r_t,w_t)): \quad &&\pi_t\equiv \ma_{ K_t,L_t\geq 0} \big(F(K_t,L_t)-r_tK_t-w_tL_t\big),
\end{eqnarray}
%$(\theta^i_t)_{i=1}^m$ is the share of profit at date $t$. $\theta_i\equiv (\theta^i_t)_t$ is exogenous, $\theta^i_t\geq 0$ $\forall i$ and $\sum_{i=1}^m\theta^i_t=1$. 
where the function $F$ satisfies Assumption \ref{Assumption1_production}.
%Assume that the function $F$ is constant return to scale, which implies the zero profit $\pi$. As above, we define the function $f: \rr_+\to\rr_+$ by $f(k)\equiv F(k,1)$ $\forall k\geq 0$.

 Each household $i$ has an endowment $e_{i,t}\geq 0$ units of consumption good and supplies $L_{i,t}\geq 0$ units of labor supply at each date $t$.\footnote{\cite{BeckerDubeyMitra2014} consider the case $L_{i,t}=1/m$.} %Let us normalize by assuming that $\sum_{i}L_{i,t}=1$.
 
Households invest in physical capital and/or financial assets and consume. In each period $t$, agent $i$ consumes $c_{i,t}$ units of consumption good. If agent $i$ buys $k_{i,t+1}\geq 0$ units of capital in period $t$, she will receive $(1-\delta)k_{i,t+1}$ units of old capital in period $t+1$, after being depreciated ($\delta$ is the depreciation rate), and $r_{t+1}k_{i,t+1}$ units of consumption good. 

As in the OLG model in Section \ref{OLG-1}, there are the so-called a pure bubble asset and a long-lived asset that brings dividends (Lucas' tree). Each household  $i$ takes the sequence $(q,p,r,w)\equiv (q_t,p_t,r_t,w_t)_{t\geq 0}$ as given and chooses the sequences of capital $(k_{i,t})_{t\geq 0}$, of the long-lived asset $(a_{i,t})_{t\geq 0}$, of fiat money $(b_{i,t})_{t\geq 0}$ and of consumption $(c_{i,t})_{t\geq 0}$  in order to maximize her intertemporal utility.
\begin{align} (P_i(q,p,r,w))&:\ma_{(c_{i,t},k_{i,t+1}, a_{i,t},b_{i,t})_{t=0}^{+\infty}} \Big[\summ_{t=0}^{+\infty} \beta_i^tu_i(c_{i,t})\Big]
\end{align}
subject to constraints $k_{i,t+1}, a_{i,t}, b_{i,t}\geq 0$,\footnote{We may eventually introduce a short-sale constraint as in \cite{LeVanPham2016}, \cite{BosiLevanPham2022} but it is not the main aim of the present paper.} and budget constraint
\begin{align} &c_{i,t}+k_{i,t+1}-(1-\delta)k_{i,t}+q_ta_{i,t}+p_tb_{i,t}\notag\\
&\leq r_tk_{i,t}+(q_t+d_t)a_{i,t-1}+p_tb_{i,t-1}+w_tL_{i,t}+e_{i,t}.
%\label{4.borrowing}\text{borrowing constraint:}
%&&(q_{t+1}+p_{t+1}d_{t+1})a_{i,t}\geq -f^i\big[p_{t+1}(1-\delta)+r_{t+1}\big]k_{i,t+1},
\end{align}
%where $e_{i,t}\geq 0$ is endowment of agent $i$ at date $t$.
%where $f^i\in [0,1]$ is borrowing limit of agent $i$. $f^i$ is an exogenous parameter and set by law.  Agent $i$ can borrow an amount but the repayment of this amount does not exceed a fraction of the market value of his physical capital. This fraction, $f^i$, is less than $1$, i.e., the market value of collateral of each agent is greater than its debt. We can prove that borrowing constraint (\ref{4.borrowing}) is equivalent to $q_ta_{i,t}\geq -f^ip_tk_{i,t+1}$.\footnote{See Remark \ref{borrowing_equivalent}.}

Denote $\ee_{GEILA}$ the economy characterized by a list $$\ee_{GEILA}=\Big((u_i,\beta_i,(e_{i,t},L_{i,t})_t,k_{i,0},a_{i,-1}, b_{i,-1})_{i=1}^m,f,(d_t)_{t},\delta\Big).$$

\begin{definition}\label{4.definition_equilibrium} A sequence of prices and quantities \\
$\big(\bar{q}_t,\bar{p}_t,\bar{r}_t, \bar{w}_t,(\bar{c}_{i,t},\bar{k}_{i,t+1}, \bar{a}_{i,t},\bar{b}_{i,t})^m_{i=1}, \bar{K}_t,\bar{L}_t\big)_{t\geq 0}$ is an intertemporal equilibrium of the economy $\ee_{GEILA}$ if the following conditions are satisfied.
(i) Price positivity: $\bar{w}_t,\bar{q}_t, \bar{r}_t>0, \bar{p}_t\geq 0$  $\forall t\geq 0$. (ii) Market clearing: $\bar{K}_t=\summ_{i=1}^m\bar{k}_{i,t}$, $\bar{L}_t=\sum_{i=1}^m L_{i,t}$, $\summ_{i=1}^m\bar{a}_{i,t}=1$,  $\summ_{i=1}^m\bar{b}_{i,t}=1$, and $$\summ_{i=1}^m(\bar{c}_{i,t}+\bar{k}_{i,t+1}-(1-\delta)\bar{k}_{i,t})=e_t+f(\bar{K}_t)+d_t \text{ } \forall t\geq 0, $$
where $e_t\equiv \sum_{i=1}^me_{i,t}$ is the aggregate endowment; 
  (iii) Optimal consumption plans: for all $i$, $(\bar{c}_{i,t},\bar{k}_{i,t+1},\bar{a}_{i,t},\bar{b}_{i,t})_{t=0}^{\infty}$ is a solution to the problem $(P_i(\bar{q},\bar{p},\bar{r}, \bar{w}))$. (iv) Optimal production plan: for all $t\geq  0$, $(\bar{K}_t,\bar{L}_t)$ is a solution to the problem $(P(\bar{r}_t,\bar{w}_t))$.

\end{definition}

%Let the functions $F$ and $f$ satisfy Assumption \ref{Assumption1}. 
We impose the standard assumptions on the households' characteristics.
\begin{assum}[endowments]
\label{Assumption2}
 (1)  $k_{i,0}, a_{i,-1},  b_{i,-1}, e_{i,t}, L_{i,t}\geq 0$, and $(k_{i,0}, a_{i,-1},e_{i,0})\not=(0,0,0)$  $\forall i\in \{1,\ldots,m\}$. Moreover,  $\sum_{i=1}^mL_{i,t}=1$, $\sum_{i=1}^m a_{i,-1}=1$, $\sum_{i=1}^m b_{i,-1}=1$. %and $K_0\equiv \sum_{i=1}^m k_{i,0}\geq 0$.

%(2) For each agent $i$, the function $u_i: \rr_+\to \rr\cup \{-\infty\}$ is concave, strictly increasing, continuously differentiable and  $u'_i(0)=+\infty$.

%We also assume that \begin{align*}\ma_{c_{i,t},k_{i,t}\geq 0} \Big\{\sum_{t=0}^{\infty}\beta_i^tu_i(c_{i,t}): c_{i,t}+k_{i,t+1}-(1-\delta)k_{i,t}\leq F_{i,t}(k_{i,t-1})\Big\}&>-\infty.\end{align*}
%PROBLEM HERE
\end{assum}
\begin{assum}\label{Assumption3}For all $i$, $\beta_i\in (0,1)$ and the utility function $u_i:\rr_+\to \rr\cup\{-\infty\}$ satisfies  (1) $\summ_{t=0}^{\infty}\beta_i^tu_i(W_t)<\infty$, where $(W_t)_t$ is defined by $W_0\equiv f(K_0)+d_0+\summ_{i=1}^me_{i,0}$ and $W_t=f(W_{t-1})+d_t+\summ_{i=1}^me_{i,t} $ $\forall t\geq 1$ and 
(2)  there exist $\theta,x\in \rr$ such that $\frac{u_i(c)-u_i(\lambda c)}{1-\lambda}\leq \theta u_i(c) +x$ $\forall \lambda\in (\underline{\lambda}, 1), \forall c\in \{z: u_i(z)>-\infty\}$, where $\underline{\lambda}\in (0,1)$.
\end{assum}

Assumption \ref{Assumption3}.(2) is a variant of Assumption (ii) in Proposition 5.1 in \cite{EkelandScheinkman1986}, which plays an important role in proving transversality conditions. This assumption is satisfied under standard setups, for instance, $u_i(c)=c^{1-\sigma}/(1-\sigma)$, where $0<\sigma\not=1$ or $u_i(c)=\ln(c)$. It also holds when $u_i(0)>-\infty$. Indeed, by the concavity of $u_i$, we have $u_i(\lambda c)\geq \lambda u_i(c)+(1-\lambda)u_i(0)$, which implies  $\frac{u_i(c)-u_i(\lambda c)}{1-\lambda}\leq u_i(c)-u_i(0)$ $\forall \lambda \in (0,1)$. Then, we take $\theta=1$ and $x=-u_i(0)$.

\begin{remark}Under Assumptions  \ref{Assumption1_production} and \ref{Assumption2}, we have $L_t=1$, $r_t=f'(K_t)$ and $w_t= f(K_t)-f'(K_t)K_t$ in equilibrium. Hence, we also call 
$\big(q_t, p_t, (c_{i,t}, k_{i,t+1},b_{i,t}, a_{i,t}) _{i\in I}, K_t  \big)_{t\geq 0}$ an intertemporal equilibrium. \end{remark}

%Let $\big(q_t, r_t,(c_{i,t}, a_{i,t}, k_{i,t})_{i=1}^m, K_t\big)_t$ be an equilibrium.  FOCs imply that, for all $t$,
%\begin{align}
%\dfrac{q_t}{q_{t+1}+d_{t+1}}=\ma_{i}\big\{\dfrac{\lambda_{i,t+1}}{\lambda_{i,t}} \big\}\leq \dfrac{1}{r_{t+1}+1-\delta}.\end{align}
%Moreover, the equality holds if $K_{t+1}>0$.

%\begin{proof}
%Since $\summ_{i=1}^ma_{i,t}=1$, there exists $i$ such that $a_{i,t}>0$, and then $\mu_{i,t+1}=0$. As a consequence, we get $$\dfrac{q_t}{q_{t+1}+d_{t+1}}=\ma_{i}\big\{\dfrac{\lambda_{i,t+1}}{\lambda_{i,t}} \big\}.$$
%It is easy to see that $\dfrac{1}{r_{t+1}+1-\delta}\geq \ma_{i}\big\{\dfrac{\lambda_{i,t+1}}{\lambda_{i,t}} \big\}$. Assume that $k_{i,t+1}>0$, we have $\lambda_{i,t+1}=0$, and then 
%$$\dfrac{1}{r_{t+1}+1-\delta}=\dfrac{\lambda_{i,t+1}}{\lambda_{i,t}}.$$
%Therefore, we have $\frac{1}{r_{t+1}+1-\delta}=\ma_{i}\big\{\frac{\lambda_{i,t+1}}{\lambda_{i,t}} \big\}.$
%\end{proof}

We now introduce the notion of two-cycle economy and two-cycle equilibrium.
%\subsubsection{Two-cycle economy and two-cycle equilibrium}
\begin{definition}[two-cycle economy]\label{definition_2cycle_economy}
The economy $\ee$ is called a two-cycle economy if (1) there are $2$ consumers, called $1$ and $2$,\footnote{Some papers name {\it odd} and {\it even} agents.} with $u_i=u$, $
\beta_i=\beta\in(0,1) $ $  \forall i=\{1,2\}$, (2) their endowments are  $k_{1,0}= 0, a_{1, -1} =0, b_{1, -1} =0, k_{2,0}\geq 0, a_{2, -1} =1, b_{2, -1} =1$, %$e_{1,2t-1}=0,e_{2,2t}=0$ $\forall t\geq 0$,
and (3) their labor endowments are $ L_{1,2t}= 1,  L_{1,2t+1}= 0,
L_{2,2t}= 0,  L_{2,2t+1}= 1$ $\forall t$.
\end{definition}

Denote this two-cycle economy by $\ee_{GEILA2}\equiv \ee_{GEILA2}(u,\beta,(e_{i,t})_t,f,\delta,(d_t)_t).$

\begin{definition}
An intertemporal equilibrium  %$\big(q_t, p_t, r_t, \left (c_{i,t}, k_{i,t+1}, a_{i,t} \right ) _{i\in I}, K_t  \big) _{t}$ 
$\big(q_t, p_t, (c_{i,t}, k_{i,t+1},b_{i,t}, a_{i,t} ) _{i\in I}, K_t\big)_{t\geq 0}$ is called a two-cycle  equilibrium of the economy $\ee_{GEILA2}$ if 
\begin{subequations}\label{kallocation}
\begin{align}
 \label{7}k_{1,2t}&=a_{1,2t-1}=  b_{1,2t-1}=0, &
%c_{1,2t-1}&=(1-\delta+r_{2t-1})K_{2t-1}+q_{2t-1}+d_{2t-1}\\
k_{1,2t+1}&=K_{2t+1}, a_{1,2t}=b_{1,2t}=1,\\
 %c_{1,2t}&=w_{2t}-K_{2t+1}-q_{2t}
\label{8} k_{2,2t}&=K_{2t}, a_{2,2t-1}= b_{2,2t-1}=1,& 
% c_{2,2t-1}&=w_{2t-1}-K_{2t}-q_{2t-1}\\
k_{2,2t+1}&= a_{2,2t}= b_{2,2t}=0.
%c_{2,2t}&=(1-\delta+r_{2t})K_{2t}+q_{2t}+d_{2t}.
\end{align}
\end{subequations}
\end{definition}
Observe that in a two-cycle equilibrium, we have
\begin{subequations}\label{callocation}
\begin{align}
% k_{1,2t}&=0, \quad a_{1,2t-1}=0\\
\label{9}c_{1,2t-1}&=e_{1,2t-1}+R_{2t-1}K_{2t-1}+q_{2t-1}+d_{2t-1}+p_{2t-1},\\
%k_{1,2t+1}&=K_{2t},\quad a_{1,2t}=1\\
c_{1,2t}&=e_{1, 2t}+w_{2t}-K_{2t+1}-q_{2t}-p_{2t},\\
 %k_{2,2t}&=K_{2t},\quad a_{2,2t}=1\\
\label{10}c_{2,2t-1}&=e_{2, 2t-1}+w_{2t-1}-K_{2t}-q_{2t-1}-p_{2t-1},\\
%k_{2,2t+1}&=0, \quad a_{2,2t}=0\\
c_{2,2t}&=e_{2,2t}+R_{2t}K_{2t}+q_{2t}+d_{2t}+p_{2t},
\end{align}\end{subequations}
where  denote $R_t\equiv r_t+1-\delta$ as in Section \ref{OLG-1}.

We have the following key result characterizing the two-cycle equilibrium.
\begin{proposition}\label{lemma-2cycle} Consider a two-cycle economy  $\ee_{GEILA2}\equiv \ee_{GEILA2}(u,\beta,(e_{i,t})_t,f,\delta,(d_t)_t).$ Let Assumptions  \ref{Assumption1}-\ref{Assumption2}  be satisfied.  Denote 
\begin{align}\label{endowment_change}
e^o_{2t}&\equiv e_{2,2t},&  e^o_{2t+1}&\equiv e_{1,2t+1},& \quad  e^y_{2t}&\equiv e_{1,2t},&  e^y_{2t+1}&\equiv e_{2,2t+1} \text{ }\forall t.
\end{align}
Let $E_t\equiv \big(q_t, p_t,  \left (c_{i,t}, k_{i,t+1},b_{i,t}, a_{i,t} \right ) _{i\in \{1,2\}}, K_t  \big) _{t\geq 0}$ be a positive list satisfying (\ref{kallocation}) and (\ref{callocation}).
\iffalse
\begin{subequations}
\begin{align}
\label{lemma2-10}& q_{t}R_{t+1} =\left(q_{t+1}+d
_{t+1}\right), \quad p_{t}R_{t+1}=p_{t+1}\\
\label{lemma2-20}&\frac{1}{R_{t+1}}=\frac{\beta u'(e^o_{t+1}+R_{t+1}K_{t+1}+q_{t+1}+d_{t+1}+p_{t+1}) }{ u'(e^y_t+w_{t}-K_{t+1}-q_{t}-p_t) }
\end{align}
\end{subequations}
\fi
\begin{enumerate}
   
\item\label{propo1_necessary}  If $E_t$ is a two-cycle equilibrium of the economy $\ee_{GEILA2}$, then, for any $t$,%(\ref{lemma2-10}-\ref{lemma2-30}) must hold. 
\begin{subequations}
\begin{align}
\label{lemma2-10}& q_{t}R_{t+1} =\left(q_{t+1}+d
_{t+1}\right), \quad p_{t}R_{t+1}=p_{t+1},\\
\label{lemma2-20}&\frac{1}{R_{t+1}}=\frac{\beta u'(e^o_{t+1}+R_{t+1}K_{t+1}+q_{t+1}+d_{t+1}+p_{t+1}) }{ u'(e^y_t+w_{t}-K_{t+1}-q_{t}-p_t) },\\
\label{lemma2-30} &\frac{1}{R_{t+1}}\geq  \frac{\beta u'(e^y_{t+1}+w_{t+1}-K_{t+2}-q_{t+1}-p_{t+1})}{u'(e^o_{t}+R_{t}K_{t}+q_{t}+d_{t}+p_t)}.
\end{align}\end{subequations}
If we require, in addition, Assumption \ref{Assumption3} and $\sum_{t=0}^{\infty}\beta^t|u(c_{i,t})|< \infty$ $\forall i\in \{1,2\}$, then the following transversality conditions hold.
\begin{subequations}
\begin{align}
\label{lemma2-40}&\limm_{t\rightarrow \infty}\beta^{2t}u'(e^y_{2t}+w_{2t}-K_{2t+1}-q_{2t}-p_{2t})(K_{2t+1}+q_{2t}+p_{2t})=0,\\
\label{lemma2-50}&\limm_{t\rightarrow \infty}\beta^{2t-1}u'(e^y_{2t-1}+w_{2t-1}-K_{2t}-q_{2t-1}-p_{2t-1})(K_{2t}+q_{2t-1}+p_{2t-1})=0.
\end{align}\end{subequations}

\item \label{propo1_sufficient} 
 $E_t$ is a two-cycle equilibrium of the economy $\ee_{GEILA2}$ if FOCs (\ref{lemma2-10}-\ref{lemma2-30}) and TVCs (\ref{lemma2-40}-\ref{lemma2-50}) hold, and $\sum_{t=0}^{\infty}\beta^{t}u(c_{i,t})\in (-\infty,\infty)$ $\forall i\in \{1,2\}$.
 
\end{enumerate}
\end{proposition}
\begin{proof}See Appendix \ref{twomodels_proof}.\end{proof}
Conditions (\ref{lemma2-10}-\ref{lemma2-30}) are first-order conditions while  (\ref{lemma2-40}-\ref{lemma2-50}) are transversality conditions. These conditions ensure that our positive list constitutes a two-cycle equilibrium. It should be noticed that we allow for $u(0)=-\infty$ and $u(c)$ may be negative.

%\begin{remark} If $u(c)=\ln(c)$, then $r_t=f'(K_t)$ and  (\ref{lemma2-1}) implies that
%\begin{subequations}
% \begin{align}\label{dynamic_1}K_{t+1}+q_{t}&=\frac{\beta}{1+\beta}\Big(f(K_t)-K_tf'(K_t)\Big)\\
%q_{t+1}+d_{t+1}&=q_t(F'(K_{t+1})+1-\delta)\\
%q_t,K_{t}&>0
%\end{align}\end{subequations}
%\end{remark}

\section{Relationship between GEILA vs OLG models}
\label{sectionmain}
We now present  our main result which shows the connection between the equilibrium in an OLG model and that in a two-cycle economy.

\begin{theorem}\label{main} Let
$\big((u_i,\beta_i,(e_{i,t},L_{i,t})_t,k_{i,0},a_{i,-1}, b_{i,-1})_{i=1}^m,(e_{t}^y,e^o_t)_t,f,\delta,(d_t)_{t}\big)$  be a list of fundamentals satisfying Assumptions  \ref{Assumption1}-\ref{Assumption2}.

\begin{enumerate}
\item\label{main_part1} (GEILA $\Rightarrow $ OLG) If $(q_t, p_t, \left (c_{i,t}, k_{i,t+1}, a_{i,t}, b_{i,t} \right ) _{i\in \{1,2\}}, K_t ) _{t\geq 0}$  is a two-cycle  equilibrium of the  economy $\ee_{GEILA2}\equiv \ee_{GEILA2}(u,\beta,(e_{i,t})_t,f,\delta,(d_t)_t)$, then the sequence $(K_{t+1},q_t,p_t)_{t\geq 0}$ is an equilibrium of the OLG economy $\ee_{OLG}\equiv \ee_{OLG}(u,\beta,(e^y_t,e^o_t)_t,f,\delta,(d_t)_t),$ where the sequence $(e^y_t,e^o_t)_{t\geq 0}$ is defined by (\ref{endowment_change}).

\item\label{main_part2} (OLG $\Rightarrow $ GEILA) Assume that a positive sequence $(q_t,p_t,K_{t+1})_{t\geq 0}$ is an equilibrium of the two-period OLG economy $\ee_{OLG}\equiv \ee_{OLG}(u,\beta,(e^y_t,e^o_t)_t,f,\delta,(d_t)_t)$. Consider a list $E_t\equiv (q_t, p_t, (c_{i,t}, k_{i,t+1}, a_{i,t}, {b}_{i,t}) _{i\in \{1,2\}}, K_t  )_{t\geq 0}$ where  $(c_{i,t}, k_{i,t+1}, a_{i,t}, b_{i,t}) _{i\in \{1,2\}}$ satisfy (\ref{kallocation}) and  (\ref{callocation}) for all $t$. 

%\begin{enumerate}    \item check\end{enumerate}

{2.(a)} $E_t$ is a two-cycle equilibrium of the  economy $\ee_{GEILA2}\equiv \ee_{GEILA2}(u,\beta,(e_{i,t})_t,f,\delta,(d_t)_t)$, where endowments $(e_{i,t})_{t\geq 0}$ are defined by (\ref{endowment_change}), if  the following conditions hold.
\begin{subequations}
\begin{align}
&\sum_{t=0}^{\infty}\beta^{t}u(c_{i,t})\in (-\infty,\infty) \text{ } \forall i\in \{1,2\},\label{finite_utility}\\
\label{add1}&\frac{1}{R_{t+1}} \geq   \frac{\beta u'(e^y_{t+1}+w_{t+1}-K_{t+2}-q_{t+1}-p_{t+1})}{u'(e^o_t+R_{t}K_{t}+q_{t}+d_{t}+p_t)} \text{ }\forall t,\\
\label{add2}&\limm_{t\rightarrow \infty}\beta^{2t}u'(c_{1,2t})(K_{2t+1}+q_{2t}+p_{2t})=0,\\
\label{add3}&\limm_{t\rightarrow \infty}\beta^{2t-1}u'(c_{2,2t-1})(K_{2t}+q_{2t-1}+p_{2t-1})=0.
\end{align}
\end{subequations}
{2.(b)} Conversely, if $E_t$ is a two-cycle equilibrium of the  economy $\ee_{GEILA2}$, then (\ref{finite_utility}) and (\ref{add1}) hold. Moreover, if we require, in addition, $\sum_{t=0}^{\infty}\beta^t|u(c_{i,t})|< \infty$ $\forall i\in \{1,2\}$ and Assumption \ref{Assumption3},  then the transversality conditions (\ref{add2}) and (\ref{add3}) hold.

%then $(q_t,K_{t+1})_t$ are asset prices and aggregate capital stocks of a two-cycle  equilibrium of the two-cycle, i.e., 
\end{enumerate}
\end{theorem}
\begin{proof}
Part 1 is a consequence of Lemma \ref{equi-olg} and Proposition \ref{lemma-2cycle}'s point \ref{propo1_necessary}. Part 2 is a consequence of Lemma \ref{equi-olg} and Proposition \ref{lemma-2cycle}'s point \ref{propo1_sufficient}. The last statement of Theorem \ref{main} follows Proposition \ref{lemma-2cycle}'s point \ref{propo1_necessary} and the transversality conditions (\ref{lemma2-40}-\ref{lemma2-50}).
\end{proof}
The intuition behind this result is the two-cycle structure of the economy $\ee_{GEILA2}$ with infinite-lived agents, which resembles the structure of the OLG economy $\ee_{OLG}$ with two-period-lived agents.

Our result leads to interesting implications. First, point \ref{propo1_necessary} shows that analyzing two-cycle equilibria requires us to understand the properties of equilibrium in a two-period OLG model.  Second, point \ref{propo1_sufficient} provides a way to construct a two-cycle equilibria from an equilibrium in a two-period OLG model.  However, we need to impose additional conditions (\ref{add1}-\ref{add3}) which are satisfied in many standard setups.

Now, let us focus on two particular cases: a pure exchange economy (i.e., there is no production) and a production economy (i.e., $e^y_t=e^o_t=e_{i,t}=0$ $\forall i,\forall t$). By applying  parts 1  and 2.(a) of Theorem \ref{main} for each of these two cases, we obtain the following results.
\begin{proposition}[exchange economy]\label{main-exchange}Let
$\big((u_i,\beta_i,(e_{i,t},a_{i,-1}, b_{i,-1})_{i=1}^m,(e_{t}^y,e^o_t)_t,(d_t)_{t}\big)$  be a list of fundamentals satisfying Assumptions  \ref{Assumption1}-\ref{Assumption2}.

\begin{enumerate}
\item (GEILA $\Rightarrow $ OLG) If $(q_t, p_t, \left (c_{i,t}, a_{i,t}, b_{i,t} \right ) _{i\in \{1,2\}}  ) _{t\geq 0}$  is a two-cycle  equilibrium of the  economy  $\ee_{GEILA2}\equiv \ee_{GEILA2}(u,\beta,(e_{1,t}, e_{2,t})_t,(d_t)_t)$, then the sequence $(q_t,p_t)_{t\geq 0}$ is an equilibrium of the OLG economy $\ee_{OLG}\equiv \ee_{OLG}(u,\beta,(e^y_t,e^o_t)_t,(d_t))$, where $(e^y_t,e^o_t)_{t\geq 0}$ is defined by (\ref{endowment_change}).

\item (OLG $\Rightarrow $ GEILA) Assume that the positive sequence $(q_t,p_t)_{t\geq 0}$ is an equilibrium of the two-period OLG economy  $\ee_{OLG}\equiv \ee_{OLG}(u,\beta,(e^y_t,e^o_t)_t,(d_t))$.\\
A list  
$(q_t, p_t,  (c_{i,t}, a_{i,t},b_{i,t}) _{i=1,2}) _{t\geq 0}$, where $ ((c_{i,t}, a_{i,t}, b_{i,t}) _{i\in \{1,2\}})_{t\geq 0}$ is given by   (\ref{kallocation}) and  (\ref{callocation}),  is a two-cycle equilibrium of the  economy $\ee_{GEILA2}\equiv \ee_{GEILA2}(u,\beta,(e_{1,t}, e_{2,t})_t,(d_t)_t)$,  where endowments $(e_{i,t})_{t\geq 0}$ are defined by (\ref{endowment_change}), if  $\sum_{t=0}^{\infty}\beta^{t}u_i(c_{i,t})\in (-\infty,\infty)$ $\forall i\in \{1,2\}$  and 
\begin{subequations}
\begin{align}
\label{add12}&u'(e^o_t+q_{t}+d_{t}+p_t) \geq   \beta R_{t+1} u'(e^y_{t+1}-q_{t+1}-p_{t+1}) \text{ }\forall t,\\
\label{add22}&\limm_{t\rightarrow \infty}\beta^{2t}u'(e^y_{2t}-q_{2t}-p_{2t})(q_{2t}+p_{2t})=0,\\
\label{add32}&\limm_{t\rightarrow \infty}\beta^{2t-1}u'(e^y_{2t-1}-q_{2t-1}-p_{2t-1})(q_{2t-1}+p_{2t-1})=0.
\end{align}
\end{subequations}
%then $(q_t,K_{t+1})_t$ are asset prices and aggregate capital stocks of a two-cycle  equilibrium of the two-cycle, i.e., 
\end{enumerate}
\end{proposition}
%\begin{proof}The first point is a direct consequence of part \ref{main_part1} of Theorem \ref{main} while the second point follows part \ref{main_part2}.a of Theorem \ref{main}\end{proof}

\begin{proposition}[production economy]\label{main-production} Let
$\big((u_i,\beta_i,(L_{i,t})_t,k_{i,0},a_{i,-1}, b_{i,-1})_{i=1}^m,f,\delta,(d_t)_{t}\big)$  be a list of fundamentals satisfying Assumptions  \ref{Assumption1}-\ref{Assumption2}.
\begin{enumerate}
\item (GEILA $\Rightarrow $ OLG) If $(q_t, p_t, \left (c_{i,t}, k_{i,t+1}, a_{i,t}, b_{i,t} \right ) _{i\in \{1,2\}}, K_t ) _{t\geq 0}$  is a two-cycle  equilibrium of the  economy economy $\ee_{GEILA2}\equiv \ee_{GEILA2}(u,\beta,f,\delta,(d_t)_t)$, then the sequence $(K_{t+1},q_t,p_t)_{t\geq 0}$ is an equilibrium of the OLG economy $\ee_{OLG}\equiv \ee_{OLG}(u,\beta,f,\delta,(d_t)_t)$.

\item (OLG $\Rightarrow $ GEILA) Assume that the positive sequence $(q_t,p_t,K_{t+1})_{t\geq 0}$ is an equilibrium of the two-period OLG economy $\ee_{OLG}\equiv \ee_{OLG}(u,\beta,f,\delta,(d_t)_t)$.   A list\\
$(q_t, p_t,  (c_{i,t}, k_{i,t+1}, a_{i,t},b_{i,t}) _{i\in\{1,2\}}, K_t) _{t \geq 0},$ 
where  $ (c_{i,t}, k_{i,t+1}, a_{i,t}, b_{i,t}) _{i\in\{1,2\}}$ is determined by  (\ref{kallocation}) and  (\ref{callocation}), is a two-cycle equilibrium of the  economy $\ee_{GEILA2}\equiv \ee_{GEILA2}(u,\beta,f,\delta,(d_t)_t)$ if  $\sum_{t=0}^{\infty}\beta^{t}u_i(c_{i,t})\in (-\infty,\infty)$ $\forall i\in \{1,2\}$ and 
\begin{subequations}
\begin{align}
\label{add12}&u'(R_{t}K_{t}+q_{t}+d_{t}+p_t) \geq   \beta R_{t+1} u'(w_{t+1}-K_{t+2}-q_{t+1}-p_{t+1}) \text{ }\forall t,\\
\label{add22}&\limm_{t\rightarrow \infty}\beta^{2t}u'(w_{2t}-K_{2t+1}-q_{2t}-p_{2t})(K_{2t+1}+q_{2t}+p_{2t})=0,\\
\label{add32}&\limm_{t\rightarrow \infty}\beta^{2t-1}u'(w_{2t-1}-K_{2t}-q_{2t-1}-p_{2t-1})(K_{2t}+q_{2t-1}+p_{2t-1})=0.
\end{align}
\end{subequations}
%where $w_{t}\equiv f(K_{t})-K_{t}f^{\prime }(K_{t})$.
%then $(q_t,K_{t+1})_t$ are asset prices and aggregate capital stocks of a two-cycle  equilibrium of the two-cycle, i.e., 
\end{enumerate}
\end{proposition}
\section{Applications: indeterminacy and asset price bubbles}
\label{appli}
In this section, we present some applications of our results for studying the issue of indeterminacy and asset price bubble. First, we provide a formal definition of asset price bubble  \citep{tirole82, tirole85, Kocherlakota1992, SantosWoodford1997,HuangWerner2000, LeVanPham2016}. Assume that we have an asset pricing equation 
\begin{align}
\label{14}
q_{t}=\frac{q_{t+1}+d
_{t+1}}{R_{t+1}}.
\end{align}
Solving recursively (\ref{14}), we obtain an asset price decomposition in
two parts
\begin{equation}\label{decomposition}
q_{t}=Q_{t,t+\tau }q_{t+\tau }+\sum_{s=1}^{\tau }Q_{t,t+s}d_{t+s},%
\text{ where }Q_{t,t+s}\equiv \dfrac{1}{R_{t+1}\ldots R_{t+s}}
\end{equation}%
is the discount factor of the economy from date $t$ to $t+s$.

\begin{definition}
\label{60}

\begin{enumerate}
\item The fundamental value of $1$ unit of asset at date $t$ is the sum of
discounted values of future dividends:
\begin{align}\label{FVdef}
FV_{t}\equiv \sum_{s=1}^{\infty }Q_{t,t+s}d_{t+s}.
\end{align}

\item We say that there is a bubble at date $t$ if $q_{t}>FV_{t}$.

\item When $d_{t}=0$ for any $t\geq 0$ (the Fundamental Value is
zero), we say that there is a pure bubble  if $q_t>0$ for any $t$ (or the fiat money's price is strictly positive).
\end{enumerate}
\end{definition}
\begin{lemma}[\cite{Montrucchio2004}, Proposition 7]\label{mont}
Consider the case $d_t>0$ $\forall t\geq 0$.  There is a bubble if and only if $\sum_{t=1}^{\infty}\frac{d_t}{q_t}<\infty$.
\end{lemma}

Letting $\tau$ in (\ref{decomposition}) tend to infinity and  using (\ref{FVdef}), we obtain $q_{t}=FV_{t}+\lim_{\tau \rightarrow \infty }Q_{t,t+\tau}q_{t+\tau }$. Thus, $q_{t}-FV_{t}>0$ if and only if $q_{0}-FV_{0}>0$. Therefore, if a
bubble exists at date $0$, it exists forever. Moreover, we also see that 
 ${q_{t+1}-FV_{t+1}}=R_{t+1}(q_{t}-FV_{t})$ $\forall t$.

We now apply our results in Section  \ref{sectionmain} to study the issue of rational asset prices and equilibrium indeterminacy. 
\subsection{Exchange economy}
First, we focus on the exchange economy. Let us define the sequence $(R_t)_{t\geq 1}$ by
\begin{align}
     \label{lemma2-2}\frac{1}{R_{t+1}}\equiv &\frac{\beta u'(e^o_{t+1}+q_{t+1}+d_{t+1}+p_{t+1}) }{ u'(e^y_t-q_{t}-p_t) } \text{ } \forall t\geq 0,
\end{align}
and summarize our equilibrium system in Proposition \ref{main-exchange} as follows:
\begin{subequations}
\begin{align}
\label{lemma2-1} q_{t}R_{t+1} &=\left(q_{t+1}+d
_{t+1}\right), \quad p_{t}R_{t+1}=p_{t+1}\\
\label{lemma2-3}\frac{1}{R_{t+1}}&\geq  \frac{\beta u'(e^y_{t+1}-q_{t+1}-p_{t+1})}{u'(e^o_{t}+q_{t}+d_{t}+p_t)},\\
\label{lemma2-4}&\limm_{t\rightarrow \infty}\beta^{2t}u'(e^y_{2t}-q_{2t}-p_{2t})(q_{2t}+p_{2t})=0,\\
\label{lemma2-5}&\limm_{t\rightarrow \infty}\beta^{2t-1}u'(e^y_{2t-1}-q_{2t-1}-p_{2t-1})(q_{2t-1}+p_{2t-1})=0.
\end{align}
\end{subequations}
According to Proposition \ref{main-exchange},  condition (\ref{lemma2-1})  is used to characterize the intertemporal equilibrium in an OLG model. Moreover, all conditions (\ref{lemma2-1}-\ref{lemma2-5}) characterize the two-cycle  equilibrium of the  economy $\ee_{GEILA2}(u,\beta,(e_{1,t}, e_{2,t})_t,(d_t)_t)$. 

We will use the  system (\ref{lemma2-1}-\ref{lemma2-5}) to show that equilibrium indeterminacy and asset price bubbles can exist along a two-cycle equilibrium.\footnote{Solving the non-autonomous system (\ref{lemma2-1}-\ref{lemma2-5}) is far from trivial (see \cite{BosiLevanPham2022}'s Section 4,  \cite{HiranoTodaJPE2025}'s Section IV and \cite{BosiLevanPham2025} for detailed analyses in the case $p_t=0 \text{ }\forall t$).}

\begin{ex}[unique equilibrium with or without bubble]\label{example1}
Assume that $u(c)=\ln(c)$ $ \forall c$, and $e^o_{t}=0$ $\forall t$. Consider a particular case where there is no fiat money (i.e., $p_t=0$ $\forall t$). In this case, condition (\ref{lemma2-1}) implies that there is a unique equilibrium in the OLG model. Moreover, the asset price is $q_t=\frac{\beta}{1+\beta}e^y_t$. This is also part of a two-cycle equilibrium in the economy  $ \ee_{GEILA2}(u,\beta,(e_{1,t}, e_{2,t})_t,(d_t)_t)$ because  FOCs and TCVs (\ref{lemma2-1}-\ref{lemma2-5}) hold.

According to Lemma \ref{mont}, the equilibrium is bubbly if and only if $\sum_{t}{d_t}/{q_t}<\infty$, or, equivalently, 
$\sum_{t}{d_t}/{e^y_t}<\infty.$\footnote{A key condition for the existence of bubble $\sum_{t}{d_t}/{e^y_t}<\infty$ is also appeared in Section 9.3.2 in \cite{BosiLevanPham2017b},  Section 5.1.1 and Section 5.2 in  \cite{BosiLevanPham2018}, Example 5 in \cite{BosiLevanPham2021wp}, and Proposition 1 in \cite{HiranoTodaJPE2025}.} In words, this requires that the dividend would be very small with respect to the endowment of the economy.\footnote{\cite{BosiLevanPham2022}'s Proposition 7 focuses on the case $q_t>0, p_t=0 $ $\forall t$, and provide conditions under which there exists a continuum equilibria of the long-lived asset. Note that their analyses still apply for the case with only fiat money (their Section 4.1.1.)} 

\end{ex}
% Note that the interest rate is given by
%\begin{align*}
%\frac{1}{R_{t+1}}=\beta\frac{e^y_t-q_{t}-p_t}{q_{t+1}+d_{t+1}+p_{t+1}}=\frac{q_t}{q_{t+1}+d_{t+1}}=\frac{\frac{\beta}{1+\beta}e^y_t}{\frac{\beta}{1+\beta}e^y_{t+1}+d_{t+1}}.
%\end{align*}

%\begin{remark}Note that in this case $\frac{\beta u'(e^o_{t+1})}{ u'(e^y_t)}=\infty$ or $R_{t+1}^*=0$ where $R_{t+1}^*$ is defined by $1=R_{t+1}^*\frac{\beta u'(e^o_{t+1})}{ u'(e^y_t)}.$
%\end{remark}

We now consider a case where the fiat money may have the strictly positive price (that is, $p_t>0$ $\forall t$).\footnote{See also \cite{weil90} for a detailed analysis of fiat money in a stochastic OLG model.}

\begin{ex}[continuum of equilibria with fiat money]
\label{exchange2}Consider an economy with only fiat money (that is $q_t=d_t=0$ for any $t$). Assume that $u'(c)=c^{-\sigma},$ where $\sigma>0$. Assume also that $e^y_t>e^o_t$ $\forall t$ and  $\lim_{t\to\infty}\beta^t(e^y_t)^{1-\sigma}=0$.

Any sequence $(p_t)_{t\geq 0}$ satisfying the following system 
\begin{align}\label{22}
e^y_t-e^o_t&\geq 2p_t\geq 0,\quad p_t=\beta p_{t+1}\Big(\frac{e^y_t-p_t}{e^o_{t+1}+p_{t+1}}\Big)^{\sigma},
\end{align}
is a sequence of prices of a two-cycle equilibrium of the  economy $ \ee_{GEILA2}(u,\beta,(e_{1,t}, e_{2,t})_t,(d_t)_t)$, where the endowments $(e_{i,t})_{t\geq 0}$ is defined by (\ref{endowment_change}).
\end{ex}
\begin{proof}See Appendix \ref{appli_proof}.\end{proof}

Let us consider two particular cases of Example \ref{exchange2}.
\begin{enumerate}
\item\label{point1_particular} Observe that $p_t=0$ $\forall t$ is a solution of the system (\ref{22}). This is a no trade equilibrium.
\item \label{case2}Focus on the case where $e^y_t=ye^t,e^o_t=de^t$ where $y,d,e>0$ where $d<y$ and $\beta e^{1-\sigma}<1$ (to ensure that $e^y_t>e^o_t$ $\forall t$ and  $\lim_{t\to\infty}\beta^t(e^y_t)^{1-\sigma}=0$). Assume that  
\begin{align}
1<\beta e (\frac{y}{de})^{\sigma}<(\frac{y}{d})^{\sigma}.
\end{align}
%$y>d, 1<\beta e (\frac{y}{de})^{\sigma}$, and $\beta e^{1-\sigma}<1$.

 Let $p$ be determined by $1=\beta e (\frac{y-p}{(d+p)e})^{\sigma}$. Then the sequence $(p_t)$ defined by $p_t=pe^t,\forall t\geq 0,$ is a two-cycle equilibrium. In this equilibrium, the fiat money's price is  strictly positive.

By combining with point \ref{point1_particular}, we observe that two sequences ($(p_t)_{t\geq 0}$ with $p_t=0,\forall t$,  and $(pe^t)_{t\geq 0}$) are two solutions to the system (\ref{22}). By using the same argument in the proof of Proposition 5 in \cite{BosiLevanPham2022}, we can prove that any sequence $(p_t)_{t\geq 0}$ defined by $0<p_0<p$ and $p_t=\beta p_{t+1}\Big(\frac{e^y_t-p_t}{e^o_{t+1}+p_{t+1}}\Big)^{\sigma}$ $\forall t$, is a solution to the system (\ref{22}). Consequently, there exists a continuum of two-cycle equilibria in which the price of fiat money is strictly positive.%\footnote{Section 4.1.1 in \cite{BosiLevanPham2022} for a full characterization in the case $\sigma=1$.} %This is an added value with respect to Example 1 in  \cite{Kocherlakota1992} where he only provides 1 equilibrium.
\end{enumerate}
\begin{remark}
Example 1 in \cite{Kocherlakota1992} is a special case of our Example \ref{exchange2} with $\sigma=2,\beta=7/8, e=8/7, p=14, y=70, d=35$. An added value with respect to Example 1 in \cite{Kocherlakota1992} is that we show a continuum of two-cycle equilibria whose fiat money's price is strictly positive while \cite{Kocherlakota1992} only presents one equilibrium.
\end{remark}

\subsection{Production economy with financial assets}
Applying Proposition \ref{main-production} for a particular where $u(c)=\ln(c) $ $\forall c>0$,  we obtain the following result.
\begin{corollary} \label{coro1}Let $u(c)=\ln(c)$ $\forall c>0$ and  $\beta\in (0,1)$.  Assume that there is no endowment, i.e., $e_{i,t}=0$ $\forall i, \forall t$. Assume that $(q_t,p_t,K_{t+1})_{t\geq 0}$ is an equilibrium of the two-period OLG economy $\ee_{OLG}\equiv \ee_{OLG}(u,\beta,f,\delta,(d_t)_t)$,  i.e., \begin{subequations}\label{systemsolved}
\begin{align}
% u'(e_{t}^{y}+f(K_{t})-K_{t}f^{\prime }(K_{t})-K_{t+1}-q_{t}-p_t)&=\beta R_{t+1}u'\big(e_{t}^{o}+R_{t+1}(K_{t+1}+q_t+p_t)\big)\\
%R_{t+1}(K_{t+1}+q_t+p_t)&=\beta R_{t+1}\Big(f(K_{t})-K_{t}f^{\prime }(K_{t})-K_{t+1}-q_{t}-p_t\Big)\\
\label{d12}K_{t+1}+q_t+p_t&=\frac{\beta}{1+\beta}w_t=\frac{\beta}{1+\beta}\big(f(K_{t})-K_{t}f^{\prime }(K_{t})\big),\\
\label{b12}q_{t}R_{t+1} &=\left(q_{t+1}+d_{t+1}\right),\\
\label{p12}p_{t}R_{t+1}&=p_{t+1},\\
\label{bd12}K_{t+1}>0, q_t&\geq 0, p_t\geq 0.
\end{align}
\end{subequations}
% where (\ref{d12}) corresponds to (\ref{d1}). 

If  
\begin{align}\label{focadd}
w_{t-1}\beta^2\big(1-\delta+f'(K_t)\big)\big(1-\delta+f'(K_{t+1})\big)\leq w_{t+1} \text{ } \forall t%\\
%\limm_{t\rightarrow \infty}\beta^t\frac{w_{2t-1}}{w_{2t}} \frac{w_{2t-3}}{w_{2t-2}}\cdots \frac{w_{1}}{w_{2}} = \limm_{t\rightarrow \infty}\beta^t\frac{w_{2t}}{w_{2t+1}} \frac{w_{2t-2}}{w_{2t-1}}\cdots \frac{w_{2}}{w_{3}} =0, 
\end{align}
then $(q_t,p_t,K_{t+1})_{t\geq 0}$ are asset prices and aggregate capital stock of a two-cycle equilibrium of the two-cycle economy $\ee_{GEILA2}\equiv \ee_{GEILA2}(u,\beta,f,\delta,(d_t)_t)$.  %More precise, the list $$\big(q_t, r_t=f'(K_t),  (c_{i,t}, k_{i,t+1}, a_{i,t}) _{i=1,2}, K_t  \big) _{t}$$ where $ (c_{i,t}, k_{i,t+1}, a_{i,t}) _{i=1,2}$ is given by (\ref{7}-\ref{10}), is a two-cycle equilibrium of the two-cycle economy $\ee_2\equiv (u,\beta,f,\delta,(d_t)_t)$.
 
\end{corollary}
\begin{proof}
Under logarithmic utility function, the Euler equation (\ref{d1}) becomes (\ref{d12}). By consequence, the TVCs (\ref{add22}) and (\ref{add32}) are satisfied. Last, thanks to \eqref{d12}, condition (\ref{add12}) becomes (\ref{focadd}).\end{proof}
We now apply Corollary \ref{coro1} to construct two-cycle equilibria with bubbles in general equilibrium models  with two infinitely-lived agents $\ee_{GEILA2}\equiv \ee_{GEILA2}(u,\beta,f,\delta,(d_t)_t)$.\footnote{Providing a complete analysis of the system (\ref{systemsolved}) is  quite hard  because it is a non-autonomous two-dimensional system with infinitely many parameters, including the dividend sequence $(d_t)$. See \cite{tirole85}, \cite{BosiHaHuyLeVanPhamPham2018}, \cite{HiranoTodaJPE2025}, \cite{PhamTodaECMA,PhamToda2025} for the interplay between dividend-paying asset and capital accumulation in OLG models.} To make clear our applications, we consider two standard cases: Linear and Cobb-Douglas production functions.
\subsubsection{Cobb-Douglas production function}
The following result is an application of Corollary \ref{coro1}.
\begin{ex}[pure bubble in a model with Cobb-Douglas production function]\label{purecd}Let $u(c)=\ln(c)$, $\beta\in (0,1)$, $\delta=1$, the Cobb-Douglas production function $f(k)=Ak^{\alpha}$, where $\alpha\in (0,1)$. Let us focus on the model with only the pure bubble asset and physical capital.  %The system (\ref{systemsolved}) becomes
%\begin{subequations}\begin{align}
%\label{d123}K_{t+1}+p_t&=\frac{\beta}{1+\beta}w_t=\frac{\beta}{1+\beta}\big(f(K_{t})-K_{t}f^{\prime }(K_{t})\big)\\
%\label{p123}p_{t}R_{t+1}&=p_{t+1}\\
%\label{bd13}K_{t+1}>0,  p_t\geq 0.
%\end{align}\end{subequations}
%This system is studied in \cite{tirole85}. 

Denote $K^{\ast }$ the capital in the bubbleless steady
state, that is the steady state without pure bubble asset $K^{\ast }=\rho^{1/\left( 1-\alpha \right) }, \text{ where } \rho\equiv \gamma \alpha A.$ 

Denote $\gamma \equiv \frac{\beta }{1+\beta }\frac{1-\alpha}{\alpha
}$. Observe that $\gamma={1}/{f^{\prime }(K^{\ast})}$.

Assume that $\gamma>1$ (i.e., $f^{\prime }(K^{\ast })<1$; this is a so-called "low interest rate condition").

Then, there exists a two-cycle equilibrium with non-negligible bubble (i.e., $p_t$ does not converge to zero) in the two-cycle economy  $\ee_{GEILA2}\equiv \ee_{GEILA2}(u,\beta,f,\delta,(d_t)_t)$. In such an equilibrium, the aggregate capital and the asset price are determined by
\begin{align}
\label{akt}
 K_{t}& =(\alpha A)^{\frac{1-\alpha ^{t}}{1-\alpha }}K_{0}^{\alpha
^{t}} \text{ }\forall t\geq 1 \text{ and } \text{ }p_{t} =(\gamma-1)K_{t+1} \text{ }\forall t\geq 0.
 \end{align}
Moreover,   $
\lim_{t\rightarrow \infty }K_{t}=(\alpha A)^{1/\left( 1-\alpha \right)
}<K^{\ast }\text{ and } \lim_{t\rightarrow \infty }p_{t}=(\gamma-1)(\alpha A)^{1/\left( 1-\alpha \right) }>0.$
\end{ex}
\begin{proof}See Appendix \ref{appli_proof}.\end{proof}

In terms of implications, Example \ref{purecd} shows that a standard model with pure bubble asset as in \cite{tirole85} can be embedded in a  GEILA model. Note that under specifications in Example \ref{purecd}, as we prove in Lemma \ref{prop4} in Appendix, the sequence $(p_t,K_{t+1})_{t\geq 0}$ determined by \eqref{akt} is the unique equilibrium satisfying two conditions: (i) condition (\ref{systemsolved}) with $q_t=d_t=0$ $\forall t$, and (ii) the asset price $p_t$ does not converge to zero.

\subsubsection{Linear technology}
Let us now consider  a linear production function: $F(K,L)=AK+wL$, where $A>0$ and $w>0$ represent, respectively, the capital and labor productivities. According to Corollary \ref{coro1}, an equilibrium $(q_t,p_t,K_{t+1})_{t\geq 0}$  of the two-period OLG economy are asset  prices and aggregate capital stocks of a two-cycle equilibrium of the two-cycle economy $\ee_{GEILA2}\equiv \ee_{GEILA2}(u,\beta,f,\delta,(d_t)_t)$ if and only if $\beta (1-\delta+A)\leq 1$.\footnote{\cite{LeVanPham2016}'s Section 6.1 corresponds to this model with $p_t=0,\forall t$. This case is also related to Proposition 5 in  \cite{BosiHaHuyLeVanPhamPham2018}.}

According to (\ref{b12}) and (\ref{p12}), we can compute that \begin{align*}
p_t&=R^tp_0, \quad q_0=\sum_{s=1}^t\frac{d_s}{R^s}+\frac{q_t}{R^t}, \text{ which implies that } q_t=R^s\big(q_0-\sum_{s=1}^t\frac{d_s}{R^s}\big).
\end{align*}
%Let us denote $Q_t\equiv 1/R^t$. Then it is easy to see that $q_0=\sum_{t=1}^Td_tQ_t+q_TQ_T$. For each price $q_0$ we will find an equilibrium
To sum up, we obtain the following result.
\begin{ex}%[\cite{LeVanPham2016}, Section 6.1]
\label{ex2}
Assume that (1) $u(c)=\ln(c)$, $\beta\in (0,1)$, (2) there is no endowment, i.e., $e_{i,t}=0$ $\forall i,\forall t$, (3) $F(K,L)=AK+wL$, (4) $R\equiv  1-\delta+A\leq 1$,
\begin{align}%\beta (1-\delta+A)\leq 1\\
\frac{\beta}{1+\beta}w&>\sum_{s=1}^t\frac{d_s}{R^s} \text{ and }\frac{\beta}{1+\beta}w>R^t\big(\frac{\beta}{1+\beta}w-\sum_{s=1}^t\frac{d_s}{R^s}\big) \text{ }\forall t.
\end{align}

Then, any sequence $(q_{t},p_t,K_{t+1})_{t\geq 0}$ satisfying the following system \begin{subequations}\label{lastexample}
\begin{align}
p_0\geq 0, \quad p_{t}=R^tp_0,  \quad \sum_{s=1}^{\infty}\frac{d_{s}}{R^{s}}&\leq q_{0}<\frac{\beta}{1+\beta}w-p_0,\\
q_{t}&=R^t\left(q_{0}-\sum_{s=1}^{t}\frac{d_s}{R^{s}}\right), \label{24}\\
nk_{t+1}+q_{t}+p_t&=\frac{\beta}{1+\beta}w,
\label{241}
\end{align}\end{subequations}
is part of a two-cycle equilibrium in the two-cycle economy  $\ee_{GEILA2}$. 
 Moreover, the following statements hold.
\begin{enumerate}
\item Fiat money has a positive price if $p_0>0$. Moreover, the supremum value $\bar{p}_0$ of initial fiat price $p_0$, which ensures that $p_t>0$ $\forall t$, is determined by $\bar{p}_0=\frac{\beta}{1+\beta}w-\sum_{s=1}^{\infty}\frac{d_{s}}{R^{s}}$.
\item If $q_0=\sum_{s=1}^{\infty}\frac{d_{s}}{R^{s}}$, then there is no bubble of the long-lived asset. In this case, we have $p_0\geq 0$. There exists a continuum of equilibria with pure bubble, indexed by $p_0$.

\item If $q_0>\sum_{s=1}^{\infty}\frac{d_{s}}{R^{s}}$, then there is a bubble of the long-lived asset. Moreover, in this case, $%
\lim_{t\rightarrow \infty }b_{t}>0$ if and only if $R=1$.

\end{enumerate}
\end{ex}
Example \ref{ex2} shows that there exists a continuum of equilibria with a strictly positive price of fiat money (pure bubble asset) and/or with bubbles of the long-lived assets.  Bubbles of the long-lived asset and fiat money can co-exist. Indeed, take $p_0>0$ so that $\sum_{s=1}^{\infty}\frac{d_{s}}{R^{s}}<\frac{\beta}{1+\beta}w-p_0$. Then, take $q_0$ so that $\sum_{s=1}^{\infty}\frac{d_{s}}{R^{s}}<q_0<\frac{\beta}{1+\beta}w-p_0$. Last, take $k_{t+1}=\frac{\beta}{1+\beta}w-q_t-p_t$. Then, the sequence $(k_{t+1},q_{t},p_t)_{t\geq 0}$ is strictly positive and satisfies  (\ref{lastexample}). By consequence, it is part of an equilibrium whose fiat money's prices are strictly positive (i.e., $p_t>0$ $\forall t$)  and the long-lived asset has a bubble (i.e., $q_0>\sum_{s=1}^{\infty}\frac{d_{s}}{R^{s}}$). %thank to the portfolio effect.

In Example \ref{ex2}, when $R<1$, we have $\lim_{t\to\infty}q_t=\lim_{t\to\infty}p_t=0$. When $R=1$, we have  $\lim_{t\to\infty}p_t=p_0$ and  $\lim_{t\to\infty}q_t=q_0-\sum_{s=1}^{\infty}\frac{d_{s}}{R^{s}}$. This shows that the growth rate and the dividend's size play an important role on the asset prices.

%We now look at  two conditions $\frac{\beta}{1+\beta}w>\sum_{s=1}^t\frac{d_s}{R^s}$ and $\frac{\beta}{1+\beta}w>R^t\big(\frac{\beta}{1+\beta}w-\sum_{s=1}^t\frac{d_s}{R^s}\big).$ 

%Let $\frac{\beta}{1+\beta}w=X+\sum_{s=1}^{\infty}\frac{d_s}{R^s}$, where $X\geq 0$.

%The latter condition is equivalent to 
%\begin{align}\frac{\beta}{1+\beta}w>R^t\big(X+\sum_{s=1}^{\infty}\frac{d_s}{R^s}\big)\\\frac{\beta}{1+\beta}w>R^t X+\sum_{s=1}^{\infty}d_sR^{t-s}\end{align}
%So, if $X>0$, we must have $R\leq 1$.

%If $X=0$, we only need \begin{align}
%\frac{\beta}{1+\beta}w=\sum_{s=1}^{\infty}\frac{d_s}{R^s}\\\frac{\beta}{1+\beta}w>R^t\sum_{s=1}^{\infty}\frac{d_s}{R^s}\end{align}
%Let $d_t=\xi^t$. Let us look at two conditions: 
%\begin{align}\frac{\beta}{1+\beta}w\geq \sum_{s=1}^{\infty}\frac{d_s}{R^s}=\frac{1}{1-\frac{\xi}{R}}\\
%\frac{\beta}{1+\beta}w>R^t(\frac{\beta}{1+\beta}w-\frac{1-(\frac{\xi}{R})^{t+1}}{1-(\frac{\xi}{R}}\big)
%\end{align}

%This requires that \begin{align}\xi<R<1.\end{align}

\section{Conclusion}
This paper bridges two foundational macroeconomic models: the infinite-horizon general equilibrium model with infinitely-lived agents (GEILA) and the overlapping generations (OLG) model. By establishing the connection between the two models, we have provided a unified view that deepens our understanding of phenomena like equilibrium indeterminacy and rational asset price bubbles in both models. Our results also allow us to construct general equilibrium models with infinitely-lived agents, where asset price bubbles exist. Moreover, we have shown that a cycle of exogenous parameters, which generates a two-cycle economy (Definition \ref{definition_2cycle_economy}),  can create equilibrium indeterminacy and asset price bubbles (see Section \ref{appli}).
\appendix
\section{Appendix}
  \label{appen}
  \subsection{Proof of Proposition \ref{lemma-2cycle}}
\label{twomodels_proof}
 \setcounter{equation}{0} 
\numberwithin{equation}{section}

{ 
To prove Proposition \ref{lemma-2cycle}, we need the following result.

\begin{lemma}\label{lemma1} %Assume that $f^i=0$ for any $i$. 
Let Assumptions \ref{Assumption1}-\ref{Assumption2} be satisfied.\\
Part A (necessary conditions). If a sequence $(q_t, p_t, r_t, \left (c_{i,t}, k_{i,t+1}, a_{i,t},b_{i,t} \right ) _{i\in I}, K_t) _{t\geq 0}$ is an equilibrium, then  there exists 
non-negative sequences $((\lambda_{i,t},\sigma _{i,t}, \mu_{i,t}, \nu _{i,t})_{i\in I})_{t\geq 0}$ satisfying the following conditions for any $t,i$:
 \begin{itemize}
 \item[(i)] $c_{i,t}>0, k_{i,t+1}\geqslant 0, \; a_{i,t} \geqslant 0, b_{i,t} \geq 0$,  $K_t \geqslant 0, q_t >0, r_t >0, p_t\geq 0$.
 \item[(ii)] $K_t= \sum _{i\in I} k_{i,t}$, $\sum _{i\in I} a_{i,t} =1$,  $\sum _{i\in I} b_{i,t} =1$.
\item[(iii)] $f (K_t)-r_t K _t = w_t=\max \{f (K)-r_tK : k\geqslant 0 \}$.
\item[(iv)] $c_{i,t}+k_{i,t+1}-(1-\delta)k_{i,t}+q_ta_{i,t}+p_tb_{i,t}
= r_t k_{i,t}+(q_t+d_t)a_{i,t-1}+p_tb_{i,t-1}+L_{i,t}w_t+e_{i,t}$.

\item[(v)] First order conditions:
 \begin{align}
 \label{fock}\lambda_{i,t}=\beta_i^tu_i'(c_{i,t}), \quad 
 \lambda_{i,t}&\geq R_{t+1} \lambda_{i,t+1}+\sigma_{i,t}, \quad \sigma_{i,t} k_{i,t+1}=0,\\
%\frac{1}{r_{t+1}+1-\delta} &= \frac{\beta_i u_i'(c_{i,t+1}) }{ u_i'(c_{i,t}) }+\sigma_{i,t}, \quad \sigma_{i,t} k_{i,t+1}=0\notag \\
% \frac{q_t}{q_{t+1}+d_{t+1}} &=\frac{\beta_i u_i'(c_{i,t+1}) }{ u_i'(c_{i,t}) } +\mu_{i,t}, \quad \mu_{i,t}a_{i,t}=0\notag\\
\label{focq} \lambda_{i,t}q_t&=(q_{t+1}+d_{t+1})\lambda_{i,t+1}+\mu_{i,t}, \quad \mu_{i,t}a_{i,t}=0,\\
  %p_t&=\frac{\beta_i u_i'(c_{i,t+1}) }{ u_i'(c_{i,t}) }p_{t+1} +\nu_{i,t}, \quad \nu_{i,t}b_{i,t}=0.\notag
  \lambda_{i,t}p_t&=\lambda_{i,t+1}p_{t+1} +\nu_{i,t}, \quad \nu_{i,t}b_{i,t}=0.\label{focp}
\end{align}
%\end{align}
\end{itemize}
If we require, in addition, Assumption \ref{Assumption3} and $\sum_{t=0}^{\infty}\beta^t|u(c_{i,t})|< \infty$, 
then we have 
\begin{align}\text{ (vi) transversality conditions: }
\limm_{t\rightarrow \infty}\beta_i^tu_i'(c_{i,t})(k_{i,t+1}+q_ta_{i,t}+p_tb_{i,t})=0. \label{tvc_proof}
\end{align}
 Part B (sufficient conditions).  If sequences $(q_t, p_t, \left (c_{i,t}, k_{i,t+1}, a_{i,t},b_{i,t} \right ) _{i\in I}, K_t) _{t\geq }$ and \\ $((\lambda_{i,t},\sigma _{i,t}, \mu_{i,t}, \nu _{i,t})_{i\in I})_{t\geq 0}$ satisfy conditions (i-vi) above, then $(q_t, p_t, \left (c_{i,t}, k_{i,t+1}, a_{i,t},b_{i,t} \right ) _{i\in I}, K_t) _{t\geq 0}$ is an intertemporal equilibrium. 
\end{lemma}

\begin{proof}[{\bf Proof of Lemma \ref{lemma1}}] For pedagogical purposes and to make the paper self-contained, we provide an elementary proof.\\
{\bf Part B (sufficient condition)}.  We use the classic approach in the optimal control theory (see  \cite{BosiLevanPham2022} for instance). It suffices to prove the optimality of the allocation $\left (c_{i,t}, k_{i,t+1}, a_{i,t},b_{i,t} \right)_{t\geq 0}$. Take an arbitrary feasible allocation $(c'_{i,t}, k'_{i,t+1}, a'_{i,t},b'_{i,t})_{t\geq 0}$. We need to prove that $\sum_{t=0}^{\infty}\beta_i^tu_i(c_{i,t})\geq \limsup_{T\to\infty}\sum_{t=0}^T\beta_i^tu_i(c'_{i,t})$. Without loss of generality, assume that the budget constraint is binding, i.e., $c'_{i,t}+k'_{i,t+1}+q_ta'_{i,t}+p_tb'_{i,t}= R_tk'_{i,t}+(q_t+d_t)a'_{i,t-1}+p_tb'_{i,t-1}+w_tL_{i,t}+e_{i,t}$. Denote $E_{i,t}\equiv w_tL_{i,t}+e_{i,t}$. We have 
\begin{align}\label{foc_consequence}
    \lambda_{i,t}(c'_{i,t}+k'_{i,t+1}+q_ta'_{i,t}+p_tb'_{i,t})=\lambda_{i,t}(E_{i,t}+R_tk'_{i,t}+(q_t+d_t)a'_{i,t-1}+p_tb'_{i,t-1}).
\end{align}
From the FOCs, we have  
\begin{subequations}\label{focsmultiply}
\begin{align}\lambda_{i,t}k'_{i,t+1}&=R_{t+1}\lambda_{i,t+1}k'_{i,t+1}+\sigma_{i,t}k'_{i,t+1},&\\
\lambda_{i,t}q_{i,t}a'_{i,t}&=\lambda_{i,t+1}(q_{t+1}+d_{t+1})a'_{i,t}+\mu_{i,t}a'_{i,t},&  \lambda_{i,t}p_{i,t}b'_{i,t}&=\lambda_{i,t+1}p_{t+1}b'_{i,t}+\nu_{i,t}a'_{i,t}.
\end{align}
\end{subequations}
By (\ref{foc_consequence}), we get
\begin{align*}
\lambda_{i,t}(c'_{i,t}-E_{i,t})=\lambda_{i,t}(R_tk'_{i,t}+(q_t+d_t)a'_{i,t-1}+p_tb'_{i,t-1})-\lambda_{i,t}(k'_{i,t+1}+q_ta'_{i,t}+p_tb'_{i,t}).
\end{align*}
Then, by taking the sum over $t$ and using  (\ref{focsmultiply}), we obtain
\begin{align*}
\sum_{t=0}^T\lambda_{i,t}(c'_{i,t}-E_{i,t})=&\lambda_{i,0}(R_0k'_{i,0}+(q_0+d_0)a'_{i,-1}+p_0b'_{i,-1})-\lambda_{i,T}(k'_{i,T+1}+q_ta'_{i,T}+p_tb'_{i,T})\\
&-\sum_{t=0}^{T-1}\lambda_{i,t}(\sigma_{i,t}k'_{i,t+1}+\mu_{i,t}a'_{i,t}+\nu_{i,t}b'_{i,t}).
\end{align*}
Applying this formula for the allocation $\left (c_{i,t}, k_{i,t+1}, a_{i,t},b_{i,t} \right )_{t\geq 0}$ and using $\sigma_{i,t}k_{i,t+1}=\mu_{i,t}a_{i,t}=\nu_{i,t}b_{i,t}=0$, we get 
\begin{align*}
\sum_{t=0}^T\lambda_{i,t}(c_{i,t}-E_{i,t})=&\lambda_{i,0}(R_0k_{i,0}+(q_0+d_0)a_{i,-1}+p_0b_{i,-1})-\lambda_{i,T}(k_{i,T+1}+q_ta_{i,T}+p_tb_{i,T}).
\end{align*}
Taking the difference between $\sum_{t=0}^T\lambda_{i,t}(c'_{i,t}-E_{i,t})$ and  $\sum_{t=0}^T\lambda_{i,t}(c'_{i,t}-E_{i,t})$,  we  obtain
\begin{align*}
\sum_{t=0}^T\lambda_{i,t}(c_{i,t}-c'_{i,t})=&\sum_{t=0}^{T-1}\lambda_{i,t}(\sigma_{i,t}k'_{i,t+1}+\mu_{i,t}a'_{i,t}+\nu_{i,t}b'_{i,t})+\lambda_{i,T}(k'_{i,T+1}+q_ta'_{i,T}+p_tb'_{i,T})\\
&-\lambda_{i,T}(k_{i,T+1}+q_ta_{i,T}+p_tb_{i,T})\\
\geq &-\lambda_{i,T}(k_{i,T+1}+q_ta_{i,T}+p_tb_{i,T}).
\end{align*}
Since $u_i$ is concave, we have $u_i(c_{i,t})-u_i(c'_{i,t})\geq u'_i(c_{i,t})(c_{i,t}-c'_{i,t})$. Then,
\begin{align*}
 \sum_{t=0}^T\big(\beta_i^tu_i(c_{i,t})-\beta_i^tu_i(c'_{i,t})\big)\geq&  \sum_{t=0}^T\beta_i^tu'_i(c_{i,t})(c_{i,t}-c'_{i,t})=\sum_{t=0}^T\lambda_{i,t}(c_{i,t}-c'_{i,t})\\
 \geq &-\lambda_{i,T}(k_{i,T+1}+q_ta_{i,T}+p_tb_{i,T}).
\end{align*}
Thanks to the transversality condition $\lim_{T\to\infty}\lambda_{i,T}(k_{i,T+1}+q_ta_{i,T}+p_tb_{i,T})=0$, we obtain $\sum_{t=0}^{\infty}\beta_i^tu_i (c_{i,t})\geq \limsup_{T\to\infty}\sum_{t=0}^T\beta_i^tu_i(c'_{i,t})$. We have finished our proof.

{\bf Part A (necessary condition)}. 

{\bf Step 1} (first order conditions). Let us prove the first order condition (\ref{focp}) (conditions (\ref{fock}) and (\ref{focq}) can be proved by using the same argument). To do so, it suffices to prove that (i) $\lambda_{i,t}p_t\geq \lambda_{i,t+1}p_{t+1}$ and (ii)   if $b_{i,t}>0$, then $\lambda_{i,t}p_t= \lambda_{i,t+1}p_{t+1}$.

Point (i). Fix a date $t$. Obviously, if $p_{t+1}=0$, then $\lambda_{i,t}p_t\geq \lambda_{i,t+1}p_{t+1}$.  Suppose now that $p_{t+1}>0$. Then, we have $p_{t}>0$. Indeed, if $p_t=0$, then in optimal $a_{i,t}>1$ because agent $i$ can take $a_{i,t}$ arbitrary large to get more consumption in date $t+1$ (because $p_{t+1}>0$) but he(she) does not need to pay any thing in date $t$ (because $p_t=0$). This is a contradiction because $a_{i,t}\leq \sum_{j}a_{j,t}=1$.

So, we focus on the case in which $p_t>0$ and $p_{t+1}>0$. 
%we have $(k_{i,0}, a_{i,-1},e_{i,0})\not=(0,0,0)$, $d_0>0$ 
By Assumption \ref{Assumption2} and $r_0>0$, we have $r_0k_{i,0}+(q_0+d_0)a_{i,-1}+p_0b_{i,-1}+w_0L_{i,0}+e_{i,0}>0$. Then, by Inada's condition, we have $c_{i,t}>0$ and $c_{i,t+1}>0$. Consider an allocation $(c'_{i,t}, k'_{i,t+1}, a'_{i,t},b'_{i,t})_{t\geq 0}$ defined  by $(k'_{i,\tau+1}, a'_{i,\tau})=(k_{i,\tau+1}, a_{i,\tau})$ $\forall \tau$,   $c'_{i,\tau}=(c'_{i,\tau})$ $\forall \tau\in \{t,t+1\}$, $b'_{i,\tau}=b'_{i,\tau}$ $\forall \tau \not=t$, and 
\begin{align*}
c'_{i,t}&=c_{i,t}-\epsilon,& b'_{i,t}&=b_{i,t}+\frac{\epsilon}{p_t},&c'_{i,t+1}=c_{i,t+1}+\frac{p_{t+1}}{p_t}\epsilon.
\end{align*}
where $\epsilon\in (0,c_{i,t})$. Clearly, this allocation is feasible. By the optimality of $(c_{i,t}, k_{i,t+1}, a_{i,t},b_{i,t})_{t\geq 0}$, we have 
$\beta_i^tu_i(c_{i,t})+\beta_i^{t+1}u_i(c_{i,t+1})\geq \beta_i^tu_i(c'_{i,t})+\beta_i^{t+1}u_i(c'_{i,t+1})$. This means that
\begin{align*}
    \frac{u_i(c_{i,t})-u_i(c_{i,t}-\epsilon)}{\epsilon}\geq \beta_i \frac{u_i\big(c_{i,t+1}+\frac{p_{t+1}}{p_t}\epsilon\big)-u_i(c_{i,t+1})}{\frac{p_{t+1}}{p_t}\epsilon}\frac{p_{t+1}}{p_t}.
\end{align*}
Let $\epsilon$ tend to $0$, we get $\lambda_{i,t}p_t\geq \lambda_{i,t+1}p_{t+1}$.

Point (ii). We now suppose $b_{i,t}>0$, then by doing the same argument as above, but with $\epsilon<0$ satisfying $b_{i,t}+\frac{\epsilon}{p_t}>0$ and $c_{i,t+1}+\frac{p_{t+1}}{p_t}\epsilon>0$, we get  $\lambda_{i,t}p_t\leq \lambda_{i,t+1}p_{t+1}$. So, we obtain $\lambda_{i,t}p_t= \lambda_{i,t+1}p_{t+1}$ if $b_{i,t}>0$.

{\bf Step 2} (transversality condition). We prove the transversality condition (\ref{tvc_proof}) by using the approach of \cite{EkelandScheinkman1986} and \cite{Kamihigashi2000}. %Since the FOCs, the sequences $(\lambda_{i,t}q_t)$ and $(\lambda_{i,t}p_t)$ are decreasing. So, they converge. Since $a_{i,t},b_{i,t}$ are bounded by $1$, the sequences  $(\lambda_{i,t}q_t)$ and $(\lambda_{i,t}p_t)$

Fix an agent $i$ and a date $t$. Denote $x_i\equiv (x_{i,s})_{s\geq 0}\equiv (k_{i,s+1},a_{i,s},b_{i,s})_{s\geq 0}$. For $\lambda\in (\underline{\lambda},1)$, define $x_i(\lambda)$ and $c_{i,t}(\lambda)$ by
\begin{align*}
   & x_{i}(\lambda)=(x_{i,0}, \ldots, x_{i,t-1},\lambda x_{i,t}, \lambda x_{i,t+1}, \ldots),\\
    &   c_{i,t}(\lambda)=c_{i,t}+(1-\lambda)(k_{i,t+1}+q_ta_{i,t}+p_tb_{i,t}), \quad c_{i,s}(\lambda)=c_{i,s} \text{ } \forall s<t, \quad c_{i,s}(\lambda)=\lambda c_{i,s} \text{ } \forall s> t.
\end{align*}
It is clear that $(c_{i,s}(\lambda),x_{i,s}(\lambda))_{s\geq 0}$ is a feasible allocation of the maximization problem of agent $i$. So, the optimality of $\left (c_{i,t}, k_{i,t+1}, a_{i,t},b_{i,t} \right )_{t\geq 0}$ implies that $    \limsup_{T\uparrow \infty}\sum_{s=0}^T\beta_i^s\big(u_i(c_{i,s}-u_i(c_{i,s}(\lambda))\big)\geq 0$.\footnote{Here, "$\uparrow$" means "increases to".} Then, we obtain
\begin{align*}
\beta_i^t\frac{u_i\big(c_{i,t}+(1-\lambda)(k_{i,t+1}+q_ta_{i,t}+p_tb_{i,t}\big)-u_i(c_{i,t})}{1-\lambda}\leq&   \limsup_{T\uparrow \infty}\sum_{s=t+1}^T\beta_i^s\frac{u_i(c_{i,s})-u_i(\lambda c_{i,s})}{1-\lambda}\\
\leq& \limsup_{T\uparrow \infty}\big(\sum_{s=t+1}^T\beta_i^s \theta u_i(c_{i,s})+\beta_i^s x\big),
\end{align*}
where we use Assumption \ref{Assumption3} in the last inequality.
Let $\lambda$ increasingly tend to $1$, we have
\begin{align*}
\beta_i^tu_i'(c_{i,t})(k_{i,t+1}+q_ta_{i,t}+p_tb_{i,t})\leq & \limsup_{T\uparrow \infty}\big(\sum_{s=t+1}^T\beta_i^s |\theta| |u_i(c_{i,s})|+\beta_i^s |x|\big).
\end{align*}
Let $t$ tend to infinity and use $\sum_{t=0}^{\infty}\beta_i^t|u(c_{i,t})|< \infty$, the right hand side converges to zero and hence we obtain $\limsup_{t\to \infty}\beta_i^tu_i'(c_{i,t})(k_{i,t+1}+q_ta_{i,t}+p_tb_{i,t})\leq 0$. Since $k_{i,t+1}+q_ta_{i,t}+p_tb_{i,t}\geq 0 $ $\forall t$, we have $\lim_{t\to \infty}\beta_i^tu_i'(c_{i,s})(k_{i,t+1}+q_ta_{i,t}+p_tb_{i,t})=0$.
\end{proof}
%\begin{proof}See Appendix \ref{appen}.\end{proof}
%\begin{align}
%\text{OLG: } 
%u'(f(K_{t})-K_{t}f^{\prime }(K_{t})-K_{t+1}-q_{t})&=\beta R_{t+1}u'\big(R_{t+1}(K_{t+1}+q_t)\big)\\
%q_{t} &=\frac{1}{R_{t+1}}\left( q_{t+1}+d_{t+1}\right)\\
%K_{t+1}, q_t&\geq 0.
%\end{align}

\begin{proof}[{\bf Proof of Proposition \ref{lemma-2cycle}}]
% (TO BE PUT IN APPENDIX)  
 
According to Lemma \ref{lemma1}, first order conditions become
 \begin{subequations}
 \begin{align}
 \frac{1}{r_{2t}+1-\delta} = \frac{q_{2t-1}}{q_{2t}+d_{2t}} = \frac{\beta_2 u_2'(c_{2,2t}) }{ u_2'(c_{2,2t-1}) }\geq  \frac{\beta_1 u_1'(c_{1,2t}) }{ u_1'(c_{1,2t-1}) }, \\
 \frac{1}{r_{2t+1}+1-\delta} = \frac{q_{2t}}{q_{2t+1}+d_{2t+1}} = \frac{\beta_1 u_1'(c_{1,2t+1}) }{ u_1'(c_{1,2t}) }\geq  \frac{\beta_2 u_2'(c_{2,2t+1}) }{ u_2'(c_{2,2t}) },\\
 p_{2t-1}=\frac{\beta_2 u_2'(c_{2,2t}) }{ u_2'(c_{2,2t-1}) }p_{2t}\geq \frac{\beta_1 u_1'(c_{1,2t}) }{ u_1'(c_{1,2t-1}) } p_{2t},\\
 p_{2t}=\frac{\beta_1 u_1'(c_{1,2t+1}) }{ u_1'(c_{1,2t}) }p_{2t+1}\geq  \frac{\beta_2 u_2'(c_{2,2t+1}) }{ u_2'(c_{2,2t}) }p_{2t+1}.
\end{align}
\end{subequations}
\iffalse Recall that 
\begin{subequations}
\begin{align}
c_{1,2t-1}&=e_{1,2t-1}+R_{2t-1}K_{2t-1}+q_{2t-1}+d_{2t-1}+p_{2t-1}\\
c_{1,2t}&=e_{1, 2t}+w_{2t}-K_{2t+1}-q_{2t}-p_{2t}\\
c_{2,2t-1}&=e_{2, 2t-1}+w_{2t-1}-K_{2t}-q_{2t-1}-p_{2t-1}\\
c_{2,2t}&=e_{2,2t}+R_{2t}K_{2t}+q_{2t}+d_{2t}+p_{2t},
\end{align}\end{subequations}
\fi
According to (\ref{9}-\ref{10}) and $\beta_1=\beta_2=\beta, u_1=u_2=u$, the inequalities in FOCs are rewritten as follows:
\begin{align*}\frac{\beta u'(e_{2,2t}+R_{2t}K_{2t}+q_{2t}+d_{2t}+p_{2t}) }{u'(e_{2, 2t-1}+w_{2t-1}-K_{2t}-q_{2t-1}-p_{2t-1})}\geq  \frac{\beta u'(e_{1, 2t}+w_{2t}-K_{2t+1}-q_{2t}-p_{2t}) }{ u'(e_{1,2t-1}+R_{2t-1}K_{2t-1}+q_{2t-1}+d_{2t-1}+p_{2t-1})}, \\
\frac{\beta u'(e_{1,2t+1}+R_{2t+1}K_{2t+1}+q_{2t+1}+d_{2t+1}+p_{2t+1}) }{ u'(e_{1, 2t}+w_{2t}-K_{2t+1}-q_{2t}-p_{2t}) }\geq  \frac{\beta u'(e_{2, 2t+1}+w_{2t+1}-K_{2t+2}-q_{2t+1}-p_{2t+1}) }{ u'(e_{2,2t}+R_{2t}K_{2t}+q_{2t}+d_{2t}+p_{2t}) }.
\end{align*}
Transversality conditions become
\begin{subequations}
\begin{align*}
\limm_{t\rightarrow \infty}\beta_1^{2t}u_1'(c_{1,2t})(k_{1,2t+1}+q_{2t}a_{1,2t}+p_{2t}b_{1,2t})&=0,\\
\limm_{t\rightarrow \infty}\beta_1^{2t+1}u_1'(c_{1,2t+1})(k_{1,2t+2}+q_{2t+1}a_{1,2t+1}+p_{2t+1}b_{1,2t+1})&=0,\\
\limm_{t\rightarrow \infty}\beta_2^{2t}u_2'(c_{2,2t})(k_{2,2t+1}+q_{2t}a_{2,2t}+p_{2t}b_{2,2t})&=0,\\
\limm_{t\rightarrow \infty}\beta_2^{2t+1}u_2'(c_{2,2t+1})(k_{2,2t+2}+q_{2t+1}a_{2,2t+1}+p_{2t+1}b_{2,2t+1})&=0.
\end{align*}
\end{subequations}
These are rewritten as follows:
\begin{subequations}
\begin{align}
\limm_{t\rightarrow \infty}\beta_1^{2t}u_1'(c_{1,2t})(K_{2t+1}+q_{2t}+p_{2t})=0,\\
\limm_{t\rightarrow \infty}\beta_2^{2t+1}u_2'(c_{2,2t+1})(K_{2t+2}+q_{2t+1}+p_{2t+1})=0.
\end{align}
\end{subequations}

Since $\beta_1=\beta_2=\beta, u_1=u_2=u$, TVCs become
\begin{subequations}
\begin{align}
\limm_{t\rightarrow \infty}\beta^{2t}u'(e_{1, 2t}+w_{2t}-K_{2t+1}-q_{2t}-p_{2t})(K_{2t+1}+q_{2t})=0,\\
\limm_{t\rightarrow \infty}\beta^{2t-1}u'(e_{2, 2t-1}+w_{2t-1}-K_{2t}-q_{2t-1}-p_{2t-1})(K_{2t}+q_{2t-1})=0.
\end{align}
\end{subequations}

\begin{remark}
With the notations $e^o_{2t}\equiv e_{2,2t},  e^o_{2t+1}\equiv e_{1,2t+1}$ and $e^y_{2t}\equiv e_{1,2t},  e^y_{2t+1}\equiv e_{2,2t+1}$, % and note that $q_{t+1}+d_{t+1}=q_tR_{t+1}$,  $p_{t+1}=p_tR_{t+1}$, 
 the inequalities in FOCs become 
\begin{align*}
\frac{\beta u'(e^o_{t+1}+R_{t+1}K_{t+1}+q_{t+1}+d_{t+1}+p_{t+1})}{ u'(e^y_t+w_{t}-K_{t+1}-q_{t}-p_t) }\geq  \frac{\beta u'(e^y_{t+1}+w_{t+1}-K_{t+2}-q_{t+1}-p_{t+1}) }{ u'(e^o_{t}+R_{t}K_{t}+q_{t}+d_{t})}.
\end{align*}
\end{remark}
\end{proof}
\subsection{Proofs for Section \ref{appli}}
\label{appli_proof}
\begin{proof}[{\bf Proof of Example \ref{exchange2}}]The system (\ref{lemma2-1}-\ref{lemma2-5})  becomes. 
\begin{subequations}
\begin{align}
\label{lemma2-11}  p_{t+1}&=p_{t}R_{t+1}\geq 0,\\
\label{lemma2-23}\frac{1}{R_{t+1}}\equiv &\frac{\beta u'(e^o_{t+1}+p_{t+1}) }{ u'(e^y_t-p_t) }\geq  \frac{\beta u'(e^y_{t+1}-p_{t+1})}{u'(e^o_{t}+p_t)},\\
\label{lemma2-45}&\limm_{t\rightarrow \infty}\beta^{2t}u'(e^y_{2t}-p_{2t})p_{2t}=\limm_{t\rightarrow \infty}\beta^{2t-1}u'(e^y_{2t-1}-p_{2t-1})p_{2t-1}=0.
\end{align}
\end{subequations}
Then, we can check these conditions under assumptions in Example \ref{exchange2}. For instance, let us look at (\ref{lemma2-45}). Since $u'(c)=c^{-\sigma}$, condition (\ref{lemma2-45}) is equivalent to $\lim_{t\to\infty}\beta^t(e^y_t-p_t)^{-\sigma}p_t=0$. This is satisfied because $p_t/ e^y_t\leq 1/2$ and $
\lim_{t\to\infty}\beta^t(e^y_t)^{1-\sigma}=0$ with $\sigma >0$.
\end{proof}

\begin{proof}
[{\bf Proof of Example \ref{purecd}}]

According to Corollary \ref{coro1}, it suffices to show that the sequence $(K_{t+1},p_t)_{t\geq 0}$ satisfies  the equilibrium system
\begin{align}
\begin{cases}w_{t-1}\beta^2 f'(K_t)f'(K_{t+1})\leq w_{t+1}, \\
\label{purecd_check}K_{1}+p_{0} =\frac{\beta }{1+\beta }w_0, \quad K_{t+1}+ p_{t}=\gamma \alpha AK_{t}^{\alpha } \text{ } \forall t\geq 0, \text{ where } \gamma\equiv \frac{\beta }{1+\beta }\frac{1-\alpha}{\alpha
}, \\
p_{t+1}=\alpha AK_{t+1}^{\alpha -1}p_{t}, \quad K_{t+1}>0, p_{t}\geq 0.
\end{cases}
\end{align}
It is easy to verify the last four conditions in the system \eqref{purecd_check}. Let us check the first condition. 

Let us focus on the case  2.(b) (that is, $\gamma>1$ and $p_0=\bar{b}$) in Lemma \ref{prop4} below.  We have $K_{t+1}= \alpha A K_{t}^{\alpha }$.
Since $\delta=1$, condition (\ref{focadd}) becomes
\begin{align*}
%w_{t-1}\beta^2\big(1-\delta+f'(K_t)\big)\big(1-\delta+f'(K_{t+1})\big)\leq w_{t+1} \quad \forall t\\
&w_{t-1}\beta^2f'(K_t)f'(K_{t+1})\leq w_{t+1} \\
\Leftrightarrow & (1-\alpha)AK_{t-1}^{\alpha}\beta^2 \alpha AK_t^{\alpha-1}\alpha AK_{t+1}^{\alpha-1}\leq (1-\alpha)AK_{t+1}^{\alpha} \Leftrightarrow  \beta^2 \leq \frac{K_{t+1}}{\alpha A K_t^{\alpha}} \frac{K_{t}}{\alpha A K_{t-1}^{\alpha}}= 1,
%\beta^2 A^2\alpha^2 K_{t-1}^{\alpha} K_t^{\alpha-1}\leq K_{t+1} 
\end{align*}
which is satisfied because $\beta <1$

\end{proof}
The following result completes the proof of Example \ref{purecd}.
\begin{lemma}%[{\bf solving the system (\ref{pbs})}]
\label{prop4} %Assume that $f(k)=Ak^{\alpha } $, $U(c,d)=\ln c+\beta \ln d$ with $0<\beta <1$, and $d_{t}=0$ for any $t$.

Let $K_0$ be exogenous and strictly positive. Consider the following system:
\begin{subequations}\label{pbs}
\begin{align}
K_{1}+b_{0}& =\frac{\beta }{1+\beta }w_0, \text{ } 
K_{t+1}+ p_{t} =\gamma \alpha AK_{t}^{\alpha } \text{ } \forall t\geq 0, \text{ where } \gamma\equiv \frac{\beta }{1+\beta }\frac{1-\alpha}{\alpha
},
\label{103} \\
p_{t+1}& =\alpha AK_{t+1}^{\alpha -1}p_{t},  \label{107}\\
K_{t+1}&>0, p_{t}\geq 0.
\end{align}
\end{subequations}
\begin{enumerate}
\item If $\gamma\leq 1$ (i.e., $f^{\prime }(K^{\ast })\geq 1$), the system (\ref{pbs}) 
has a unique solution given by
\begin{equation}
p_t=0, \text{ } 
K_{t}=\rho^{\frac{1-\alpha ^{t}}{1-\alpha }}K_{0}^{\alpha
^{t}} \text{ } \forall t\geq 1,  \label{110}
\end{equation}
where $\rho \equiv \gamma \alpha A$. Moreover, $\lim_{t\rightarrow \infty }K_{t}=K^{\ast }$.

\item If $\gamma>1$ (i.e., $f^{\prime }(K^{\ast })<1$), the system   (\ref{pbs}) is indeterminate: The set of solutions is any sequence  $(K_{t+1},p_{t})_{t\geq
0}$ defined by (\ref{103}), (\ref{107}), and $p_{0}\in \left[ 0,\bar{b}\right] $, where the so-called bubble critical value $\bar{b}$ is defined by
\begin{equation}\label{bx}
\bar{b}\equiv w_0 \frac{\beta }{1+\beta }\frac{%
\gamma-1}{\gamma }=(1-\alpha)AK_0^{\alpha}\left[ 1-%
\frac{1+\alpha \beta }{\left( 1-\alpha \right) \left(
1+\beta \right) }\right],
\end{equation}%
which is positive because $\gamma>1$.

Moreover, the following properties hold.
\begin{enumerate}
\item (bubbleless solution) If $p_{0}=0$, then $p_{t}=0$ $\forall t$. The sequence $\left( K_{t}\right)_{t\geq 1}$ is determined by (\ref{110}).

\item (bubbly solution) If $p_0>0$, then $p_{t}>0$ $\forall t$.

When $p_{0}<\bar{b}$, we have $\lim_{t\rightarrow \infty }p_{t}=0$ and $%
\lim_{t\rightarrow \infty }K_{t}=K^{\ast }$.

When $p_{0}=\bar{b}$, we have $\lim_{t\rightarrow \infty }p_{t}>0$. Precisely, we have
\begin{align}\label{122} 
p_{t}& =(\gamma-1)K_{t+1}\text{ } \forall t\geq 0, \quad 
K_{t} =\rho _{1}^{\frac{1-\alpha ^{t}}{1-\alpha }}K_{0}^{\alpha
^{t}}\text{ } \forall t\geq 1,
\end{align}%
where $\rho _{1}\equiv \alpha A$. Moreover, $\lim_{t\rightarrow \infty }K_{t}=(\alpha A)^{1/\left( 1-\alpha \right)
}<K^{\ast }\text{ and } \lim_{t\rightarrow \infty }p_{t}=(\gamma-1)(\alpha A)^{1/\left( 1-\alpha \right) }>0.$
\end{enumerate}
\end{enumerate}
\end{lemma}

\begin{proof}[{\bf Proof of Lemma \ref{prop4}}]The proof presented here is similar to the proof in the literature (see Proposition 4 in \cite{BosiHaHuyLeVanPhamPham2018} among others).

If $p_0>0$, then we have $p_t>0,\forall t$.
Combining (\ref{103}) and (\ref{107}), we have 
\begin{align}
\frac{K_{t+1}}{p_t}+1 =\frac{\gamma \alpha AK_{t}^{\alpha }}{ p_{t}}=\frac{\gamma \alpha AK_{t}^{\alpha }}{\alpha AK_{t}^{\alpha -1}p_{t-1}}=\gamma \frac{K_{t}}{p_{t-1}}\text{ } \forall t\geq 1.
\end{align}

Denote $z_{t}\equiv K_{t+1}/p_{t}$ $\forall t\geq 0$. We get a difference equation: $z_{t+1}=\gamma z_{t}-1 \text{ } \forall t\geq 0.$

If $\gamma\neq 1$, the solution of
this difference equation must satisfy
$$z_{t}=\gamma^{t}z_{0}-\frac{1-\gamma^{t}}{1-\gamma} \text{ } \forall t\geq 1$$

We now consider different cases.
\begin{enumerate}
\item 
When $\gamma\leq 1$, there is no bubble. Indeed, suppose that there is a pure bubble. Since $\gamma\leq 1$, condition $z_{t+1}=\gamma z_{t}-1$ implies that $z_{t}$ becomes negative soon or later: this leads to a
contradiction. In this case, capital transition becomes $k_{t+1}=\rho k_{t}^{\alpha }$, where $\rho \equiv \gamma \alpha A$. Solving recursively, we find the explicit
solution (\ref{110}). 

\item  Let $\gamma>1$. If $p_{t}=0$, then (\ref{110}) follows immediately.

If $p_{t}>0$. Then, we obtain%
\begin{equation}
z_{t}=\frac{\left[ \left( \gamma-1\right) z_{0}-1\right] \gamma
^{t}+1}{\gamma-1}  \label{x_t}.
\end{equation}%
A positive solution exists if and only if $z_{0}\geq 1/\left( \gamma-1\right) $. Hence, the existence of a positive solution requires%
\begin{equation*}
p_{0}\leq (\gamma-1)K_{1}=(\gamma-1)%
\left[ \frac{\beta }{1+\beta }w_0 -p_{0}\right].
\end{equation*}%
Solving this inequality for $p_{0}$, we find $0<p_{0}\leq \bar{b}$.

\end{enumerate}

We now observe that for $p_{0}\in \left( 0,\bar{b}\right] $ given, the sequence $\left(
K_{t+1},p_{t}\right) $ constructed by (\ref{103}) and (\ref{107}) is a solution with $p_{t}>0$ for any $t$.

When $p_{0}<\bar{b}$ (that is $z_{0}>1/\left( \gamma-1\right) $),
thanks to (\ref{x_t}), we get $\lim_{t\rightarrow \infty }z_{t}=\infty $.
According to (\ref{103}), $K_{t}$ is uniformly bounded from above, which
implies that $\lim_{t\rightarrow \infty }p_{t}=0$. Thus, $\lim_{t\rightarrow
\infty }K_{t}=K^{\ast }$.

When $p_{0}=\bar{b}$, we have $z_{t}=1/\left( \gamma-1\right) $  $\forall t\geq 0.$ In this case, $k_{t+1}=\rho _{1}k_{t}^{\alpha }$, where $\rho
_{1}\equiv \alpha A$, and $p_{t}=\left( \gamma-1\right)
k_{t+1}$. Solving recursively, we get  (\ref{122}).

%TCIMACRO{\TeXButton{End Proof}{\endproof}}%
%BeginExpansion
\end{proof}

\end{document}